\documentclass[11pt]{article}

\usepackage{amsmath, amssymb, amsthm, mathrsfs} 
\usepackage{algorithm, algpseudocode}
\usepackage{natbib}
\usepackage{graphics, graphicx} 
\usepackage[margin=1.1in]{geometry} 
\usepackage{theoremref}

\usepackage{fancyhdr} 
\usepackage{url} 
\usepackage[normalem]{ulem}
\usepackage[dvipsnames]{xcolor} 

\usepackage{setspace}

\usepackage{tikz}
\usetikzlibrary{positioning}
\usepackage{caption}

\newcommand{\prob}{\ensuremath{\mathbf{P}}}
\newcommand{\expec}{\ensuremath{\mathbf{E}}}
\newcommand{\ind}{\ensuremath{\mathbf{I}}}
\newcommand{\reals}{\ensuremath{\mathbb{R}}}
\newcommand{\naturals}{\ensuremath{\mathbb{N}}}
\newcommand{\defeq}{\ensuremath{\triangleq}}

\newcommand{\sQ}{\ensuremath{\mathsf{Q}}}

\newcommand{\lstate}[2]{\ensuremath{\bigl(Z_{#1}^{(#2)},N_{#1}^{(#2)}\bigr)}}
\newcommand{\zstate}[2]{\ensuremath{Z_{#1}^{(#2)}}}
\newcommand{\nagent}[2]{N_{#1}^{(#2)}}
\newcommand{\mminfty}[2]{\ensuremath{\mathsf{M/M/\infty}(#1,#2)}}

\newcommand{\nhat}{\widehat{n}}
\newcommand{\lhat}{{\widehat{\ell}}}
\newcommand{\that}{{\widehat{t}}}

\newcommand{\sets}{\ensuremath{\mathbb{S}}}

\newcommand{\calK}{\ensuremath{\mathcal{K}}}
\newcommand{\setZ}{\ensuremath{\mathbb{Z}}}
\newcommand{\bounded}{\ensuremath{\mathcal{C}_b(\sets)}}
\newcommand{\compact}{\ensuremath{\widehat{\markovian}}}
\newcommand{\markovian}{\ensuremath{\Pi}}
\newcommand{\signed}{\ensuremath{\mathcal{M}(\sets)}}

\newcommand{\switch}{\ensuremath{\mathsf{sw}}}
\newcommand{\stay}{\ensuremath{\mathsf{st}}}
\newcommand{\resource}{\ensuremath{\mathsf{res}}}
\newcommand{\vswitch}{\ensuremath{V_{\switch}}}
\newcommand{\varrive}{\ensuremath{V_{\mathsf{arr}}}}
\newcommand{\vstay}{\ensuremath{V_{\stay}}}
\newcommand{\valfn}{\ensuremath{\mathcal{V}}}
\newcommand{\valsw}{\ensuremath{\mathcal{V}_\switch}}
\newcommand{\valst}{\ensuremath{\mathcal{V}_\stay}}
\newcommand{\dynamic}{\ensuremath{\mathbf{T}}}
\newcommand{\opt}{\ensuremath{\mathsf{OPT}}}
\newcommand{\markov}{\ensuremath{\mathsf{MC}}}
\newcommand{\optstop}{\ensuremath{\mathsf{DEC}}}
\newcommand{\map}{\ensuremath{\mathcal{R}}}
\newcommand{\sdleq}{\ensuremath{\preccurlyeq_{\mathsf{sd}}}}
\newcommand{\lvswitch}{\ensuremath{\underline{\mathsf{V}}}}
\newcommand{\uvswitch}{\ensuremath{\overline{\mathsf{V}}}}
\newcommand{\dist}{\ensuremath{\mathsf{dist}_\map}}

\newcommand{\bx}{\ensuremath{\mathbf{x}}}
\newcommand{\by}{\ensuremath{\mathbf{y}}}

\newtheorem{theorem}{Theorem}[section]

\newtheorem{lemma}{Lemma}[section]

\newtheorem{definition}{Definition}

\usepackage{enumitem}
\setlist{leftmargin=0.55in}

\provideboolean{submissionversion}

\ifthenelse{\boolean{submissionversion}}{
  
  \newcommand{\delkri}[1]{}
  \newcommand{\editkri}[1]{}

}{
  
  \newcommand{\delkri}[1]{{\color{RoyalBlue} \sout{#1}}}
  \newcommand{\editkri}[1]{{\color{Purple} KI: #1}}
}

\title{Mean Field Equilibria for Competitive Exploration in Resource Sharing Settings}
\date{\today}
\author{Pu Yang, Krishnamurthy Iyer, Peter Frazier}

\graphicspath{ {./plots/} }

\begin{document}
\maketitle

\begin{abstract}


We study a model of competition among nomadic agents for
    time-varying and location-specific resources, arising in
    crowd-sourced transportation services, online communities, and
    traditional location-based economic activity.  This model
    comprises a group of agents and a single location endowed with a
    dynamic stochastic resource process. Periodically, each agent
    derives a reward determined by the location's resource level and
    the number of other agents there, and has to decide whether to
    stay at the location or move. Upon moving, the agent arrives at a
    different location whose dynamics are independent and identical to
    the original location.
    Using the methodology of mean field equilibrium, we study the
    equilibrium behavior of the agents as a function of the dynamics
    of the stochastic resource process and the nature of the
    competition among co-located agents.  We show that an equilibrium
    exists, where each agent decides whether to switch locations based
    only on their current location's resource level and the number of
    other agents there.  We additionally show that when an agent's
    payoff is decreasing in the number of other agents at her
    location, equilibrium strategies obey a simple threshold
    structure.  We show how to exploit this structure to compute
    equilibria numerically, and use these numerical techniques to
    study how system structure affects the agents' collective ability
    to explore their domain to find and effectively utilize
    resource-rich areas.
\end{abstract}

\section{Introduction}
\label{sec:intro}

We consider a model of nomadic agents exploring and competing for
time-varying stochastic location-specific resources.  Such multi-agent
systems arise in many real-world settings, as illustrated below.

They arise in the sharing
economy, in crowd-sourced transportation services like Uber and Lyft, 
and in crowdsourced food delivery services like GrubHub and DoorDash,
in which drivers choose neighborhoods and then earn money
based on the number of riders or eaters requesting service 
within that
neighborhood (the location-specific resource), and the number of other
drivers working there.  This overall resource level varies
stochastically as demand
rises and falls, and the resource derived by a driver decreases as 
more drivers drive in her neighborhood.

They also arise in the traditional economy, for example in mobile food
vendors deciding where to locate their trucks; in pastoralists deciding where
to graze their livestock; and in fishermen deciding where to fish.
In these examples, the level of resource derived by each agent from their
location (whether profit from hungry passers by,
or food for livestock provided by the range-land, or profit from the
catch) depends both on the number of other agents at the location, and
on the location's stochastically varying resource level.

They also arise 
in online communities like Reddit and Twitch, in which participants
choose sub-communities or channels and then derive enjoyment
depending both on some underlying but transitory societal interest in
the sub-community's topic of focus (the overall resource) and the
number of other participants in the sub-community.  When the number of
other participants is too small, lack of social interaction prevents
enjoyment; when the number of other participants is too large,
crowding diminishes the sense of community. 

They even arise among
scientific researchers, who choose a research area in which to
work and derive value based on the underlying
level of societal interest and funding in their chosen area, and in
the number of other researchers working in it. As with online
communities, the number of other researchers should be neither too
large nor too small to maximize the value derived.

In each of these examples, the overall welfare of the system is
determined by how agents explore their domain to find and
exploit resource-rich locations. This willingness to explore in turn
depends on the level of competition or co-operation among agents at
the same location, and the distribution of agents and resources across
locations.

In this paper, we develop a formal model to analyze such
spatio-temporal competition among agents and the equilibrium behavior
of such systems.  The model we study comprises a single location and a
group of agents.  This location represents one in a large collection
of locations between which the agents move.  It has a resource level
that varies stochastically with time. Each agent at the location
periodically obtains a payoff whose amount is determined by the number
of other agents currently at the location, and the location's current
resource level. Based on these quantities, the agent then decides
whether to stay at the same location or leave.  Upon leaving the agent
receives a reward that represents the expected future discounted
payoff that would be obtained by moving to another randomly chosen
location in the system.  The agents are fully strategic and seek to
maximize the total expected payoff over their lifetime.



Using the methodology of mean field equilibrium, we study the
equilibrium behavior of the agents in this system as a function of the
dynamics of the spatio-temporal resource process and the level of
competition in the agents' sharing of a location's resources.  
We prove the existence of an equilibrium for general resource-sharing
functions. For the specific case where the resource-sharing function
is non-increasing in the number of agents at the location, we further
show that the equilibrium strategy has a simple threshold structure,
in which it is optimal for an agent to leave a location when the
number of other agents there exceeds a threshold that depends on the
location's resource level.  This result enables a simple description
of equilibrium strategies, and allows us to efficiently compute an
equilibrium.


Using numerical analysis of a setting with two resource levels and
decreasing resource-sharing function, we investigate how the
equilibrium welfare depends on resource levels' rate of change and the
density of agents.  Here, the equilibrium welfare is the sum of
payoffs earned across all agents in equilibrium, normalized to the
length of time over which these payoffs have accrued and either the
number of agents or the number of locations.  Using this methodology
we show qualitatively different system behavior when the
single-location welfare function (the contribution to welfare from all
agents at one location) increases with the number of agents at the
location as compared with when it decreases.  Our ability to derive
these and other insights discussed in detail in
Section~\ref{sec:numerics2} provide evidence that our model and
equilibrium notion lend themselves to analysis through simple
numerical methods.  Specifically, our methodology presents a
  promising approach to evaluate engineering interventions, such as
  providing subsidies to or imposing costs on agents to promote or
  discourage exploration to improve welfare.


\subsection{Related Work}

Our work contributes to the literature on mean field equilibrium
\citep{sachin_2010,huang_2007,jovanovic_1988,lasry_2007,weintraub_2008},
that studies complex systems under a large system limit and obtains
insights about agent behavior that are hard to obtain from analyzing
finite models. The main insight behind this literature, that in the
large system limit agents' behavior is characterized by their private
state and an aggregate distribution of the rest of system, has been
used to study settings including industry dynamics and oligopoly
models \citep{hopenhayn_1992,weintraub_2008,weintraub_2010}, repeated
dynamics auctions \citep{balseiro2014, iyerJS2014}, online labor
markets \citep{arnosti}, queueing
\citep{manjrekarRS2014,xu2013supermarket}, content sharing
\citep{LiBPSS2017}, and pedestrian motion \citep{lachapelleW11}, among
others. In these papers, the unit of analysis is a single agent's
decision problem, assuming the behavior of all other agents together
constitutes a mean field distribution. In contrast, in our work, the
unit of analysis is the game among the agents at a single location,
assuming that the behavior of agents and the resource level at all
other locations constitutes a mean field distribution.


Our work also contributes to the literature on spatial models of
ride-sharing and crowd-sourced transportation
\citep{banerjeeFL16,banerjeeJR15,bravermanDLY16}.  In this literature,
the paper most closely related to ours is \citep{bimpikisCD16}, who
consider a ride-sharing platform with a continuum of riders and
drivers spread across a finite network of locations, and study how the
platform should set origin-based prices to maximize profits. In
particular, the drivers' decision of where, when, and whether to
provide service is explicitly modeled. The paper studies the impact
of the underlying network structure of the locations on the platform's
profits and consumers' surplus, under the assumption that the demand
at each location is stationary. In contrast, in our model, the
resources at each location (analogous to demand) are stochastic and
time varying.  However, in our model, agents
  decide whether to stay or switch from their current location, and
  not which location to switch to.

Our model is also related to congestion games
\citep{nisan2007algorithmic,rosenthal73}, in which agents choose paths
on which to travel, and then incur costs that depend on the number of
other agents that have chosen the same path.  One may view paths as
being synonymous with locations in our model, and observe that in both
cases the utility/cost derived from a path/location depends on the
number of other agents using that path, or portion thereof. The main
difference between our model and congestion games is the stochastic
time-varying nature of our overall level of resource (making our model
more complex), and the lack of interaction between locations
contrasting with the interaction between paths (making our model
simpler).

Another related strand of literature studies ecological models of
metapopulations in static and dynamic habitats
\citep{durrettL94,durrettL94b,levin70,molofsky75}. \citet{keymerMVL2000}
consider a set of habitats, arranged on a lattice, each containing a
subpopulation of a species, and where the landscape structure of each
habitat is stochastic and dynamic. Using a mean-field analysis, and
through numerical simulations, the authors study the dependence of
persistence and extinction rates of the species across habitats as a
function of the rate of change of the landscape. In such models, the
species dynamics are exogenously specified, whereas we are interested
in the equilibrium behavior of agents.

Our work can be seen as an extension of the Kolkata Paise Restaurant
Problem \citep{chakrabarti2009kolkata}, a generalization of the El
Farol bar problem \citep{arthur1994inductive,chakrabarti2007kolkata}.
In this game, each agent chooses (simultaneously) a restaurant to
visit, and earns a reward that depends both on the restaurant's
fixed rank, which is common across agents, and the number of other
agents at that restaurant.  This reward is inversely proportional to
the number of agents visiting the restaurant. The Kolkata Paise
Restaurant Problem is studied both in the one-shot and repeated
settings, with results on the limiting behavior of myopic
\citep{chakrabarti2009kolkata} and other strategies
\citep{ghosh2010statistics}, although we are not aware of existing
results on mean-field equilibria in this model. The model we consider
is both more general, in that we allow general reward functions and
allow a location's resource to vary stochastically, and more specific,
in that our locations are homogeneous.  Our model also differs in that
our agents' decisions are made asynchronously.

\section{Model}
\label{sec:model}
 
Our formal model is motivated by considering a system of locations
occupied by strategic agents.  Each location has a resource level that
varies over time according to a finite-state continuous-time Markov
chain, and is occupied by a time-varying collection of agents.  Each
agent has an associated sequence of independent decision epochs
separated by exponential times.  At the start of an agent's decision
epoch, she receives a payoff that depends on the number of other
agents at her current location and the resource level there.  She then
decides whether to stay at her location or to leave and move to
another location.  When moving, her destination is chosen uniformly at
random from the set of all locations other than her origin.  Each
agent seeks to maximize her expected payoff over an independent
exponentially distributed lifetime.  When an agent's lifetime expires,
she exits the system and is replaced by a new agent who arrives to a
new location chosen uniformly at random.  In Appendix
  \ref{ap:finite-model}, we provide a detailed description of this
  system with finite number of agents and locations.


  Since the payoff obtained by an agent at any location is determined
  by the number of agents at that location, each agent's decision to
  stay in her current location or to move to a new one depends on all
  the other agents' behavior. Consequently, the interaction among the
  agents in this finite model is a dynamic game, and describing the
  agents' behavior requires an equilibrium analysis. Since the agents
  are not fully informed about the resource levels at other locations,
  the standard equilibrium concept to analyze the induced dynamic game
  is a {\em perfect Bayesian equilibrium} (PBE). A PBE consists of a
  strategy $\xi^i$ and a belief system $\mu^i$ for each player $i$. A
  belief system $\mu^i$ for agent $i$ specifies a belief
  $\mu^i(h^i_t)$ after any history $h^i_t$ over all aspects of the
  system that she is uncertain of and that influence her expected
  payoff. A PBE then requires two conditions to hold: (1) each agent
  $i$'s strategy $\xi^i$ is a best response after any history $h^i_t$,
  given their belief system and given all other agents' strategies;
  and (2) each agent $i$'s beliefs $\mu^i(h^i_t)$ are updated via
  Bayes' rule whenever possible (see \citep{fudenbergT91,tirole} for
  more details).

  A PBE supposes a complex model of agent behavior. Each agent keeps
  track of her entire history, and maintains complex beliefs about the
  rest of the system.  While this behavioral model may be
    plausible in small settings, in large systems an agent's history
    may not contain too much information about the state of all other
    locations, since the agent would typically only visit a small
    fraction of the locations. In such settings, it is more
  plausible that each agent would base her decision to stay or switch
  solely on the current state of the location she is in ---
  specifically on its level of resource and congestion --- and on the
  aggregate features of the entire system.  Moreover, we expect that
  an agent would prefer to stay at a location with a high resource
  level and few other agents.
Below, we seek to uncover this intuitive behavioral
model as an equilibrium in large systems by letting the number of
agents and the number of location both increase proportionally to
infinity, and studying the limiting infinite system.

As the number of locations and agents grows to infinity proportionally (with the
proportionality constant $\beta>0$ defined as the {\em agent
  density}), it is reasonable to suppose that the dynamics at any
fixed finite collection of location is independent asymptotically, and
that the rewards experienced by an agent can be described by modeling
the dynamics at a single location and then supposing that upon leaving
that location the agent moves to another location whose dynamics are
independent and identically distributed, \textit{ad infinitum} until
her lifetime expires.  Thus, to analyze a large finite system, we
posit a formal model for the dynamic of a single location, and treat
each agent who leaves this location as returning to an independent
copy.
 
\subsection{Formal Model of a Single Location in the Limiting Infinite System}\label{sec:formal_model}
Here we state our formal model of a single location $k$.  Let $Z_t^k$
denote the resource level at the location at time $t \geq 0$. We
assume the resource process $\{Z_t^k : t \geq 0\}$ is a finite state
continuous time Markov chain. We let $\setZ$ denote the set of values
the resource process can take. Furthermore, we let $\mu_{zy}>0$ denote
the transition rate of $Z_t^k$ from a state $z \in \setZ$ to a state
$y \in \setZ$. We assume that the process $Z_t^k$ is irreducible and
positive recurrent, with a unique invariant distribution
$\{\pi_{\resource}(z) : z \in \setZ\}$.

We let $N_t^k$ denote the number of agents at the location $k$ at time
$t$. The stochastic process $(Z_t^k, N_t^k)$ will evolve according to
arrivals to this location, and the decisions made by agents at this
location.  Toward that end, we suppose that new agents arrive to this
location according to a Poisson process with rate $\kappa$, and we
describe the agents' decision process below.  The rate $\kappa$ models
both arrivals of agents switching from other locations in a finite
system, and new arrivals of agents to the system following the exit of
other agents from the system, but here it is taken to be an input to
the formal model of a single location, and below it is required to
satisfy consistency conditions in equilibrium.
 
Associated with each agent $i$ at location $k$ is a Poisson clock with
rate $\lambda$, such that each time the clock rings, the agent decides
whether to stay in the location or leave. We refer to each clock ring
of agent $i$ as her decision epoch, and let $\tau_i^\ell$ denote the
time of her $\ell^{th}$ decision epoch.
 
At each decision epoch $\tau_i^\ell$, the agent $i$ receives a payoff
$F(Z^k_t, N^k_t) \vert_{t=\tau_i^\ell}$ that depends on the resource
level $Z^k_t$ and the number of agents $N^k_t$. We refer to the
function $F$ as the {\em resource-sharing function}. We assume that
the resource-sharing function is non-negative, i.e., $F(z,n) \geq 0$
for each $z \in \setZ$ and $n \geq 1$. To avoid trivialities, we
require that there exists a $(z_0,n_0)$ such that $F(z_0, n_0) >
0$. Finally, to model the competitive nature of interaction among the
agents, we assume that as the number of agents at a location
increases, the payoff an agent receives approaches zero:
$\lim_{n \to \infty} F(z,n) = 0$ for each $z \in \setZ$.

To model agents with finite lifetimes, we assume that subsequent to
receiving a payoff at time $\tau_i^\ell$, with probability
$1 - \gamma \in (0,1)$, the agent's lifetime expires and the agent
exits the system permanently. Thus, each agent $i$ can exist in the
system for at most a random time interval distributed exponentially
with rate $\lambda (1 - \gamma)$. We refer to $\gamma \in (0,1)$ as
the {\em survival probability}.

If the agent's lifetime does not expire, then the agent $i$ decides
whether to stay at her location or move.  Agents are free to make this
choice based on their history of past observations.  If the agent
stays, then the dynamics and payoffs described above continue forward
for another decision epoch.  If the agent leaves, then the agent is
awarded a one-time payoff of $V_\switch>0$ and no subsequent payoffs.
Here, $V_\switch$ is taken simply to be a constant input to our model
for a single location, and below it is required to satisfy a condition
at equilibrium.  This condition corresponds to $V_\switch$ being the
conditional expected payoff experienced by an agent when moving to a
new location whose current number of agents and resource level is
distributed according to the stationary distribution induced by
equilibrium agent behavior.

\subsection{The Single-Location Decision Problem When Other Agents
  Follow Markovian Strategies}
Having specified the arrival process and agents' decision process in a
single location, we are interested in characterizing a symmetric
equilibrium among agents. For a given arrival rate $\kappa$ and the
switching payoff $\vswitch$, the particular notion of equilibrium we
consider is a Markov perfect equilibrium \cite{tirole}, where in
equilibrium, each agent finds it optimal to base her decision only on
the current state of the location at her decision epoch, and not on
her past (although she is not restricted from doing so). Formally, let
$\sets =\setZ \times \naturals_0$ denote the set of possible states of
the process $(Z_t^k,N_t^k)$.  A Markovian strategy for an agent is a
function $\xi: \sets \rightarrow [0,1]$, where $\xi(z,n)$ denotes the
probability with which the agent chooses to stay if the state of the
location at her decision epoch is $(z,n) \in \sets$. (Note that
$\xi(z,0)$ is not well-defined; by convention, we let $\xi(z,0) = 1$
for all $z \in \setZ$).
 
As a step towards formulating the game among the agents, we first
study the dynamics at a location when all agents in location $k$ adopt
a Markovian strategy $\xi$. Given the arrival rate $\kappa$ and the
Markovian strategy $\xi$, the process $(Z^k_t, N^k_t)$ for any
location $k$ evolves as a continuous time Markov chain on the state
space $\sets$ with the following transition rate matrix
$\sQ^{\xi, \kappa}$:
\begin{equation}
 \label{eq:loc-dyn}
 \begin{split}
   \sQ^{\xi, \kappa}&\left( (z,n) \to (x,m)\right) = \ind\{x\neq z, m = n\} \mu_{z,x} + \ind\{x=z, m = n+1\} \kappa \\
   & + \ind\{x=z, m = n-1\} \lambda n\left(1 - \gamma\xi(z,n)\right)\\
   & - \ind\{x=z, m=n\} \left( \sum_{y \neq z} \mu_{z,y} + \kappa +
     \lambda n \left(1 - \gamma\xi(z,n)\right)\right),
 \end{split}
\end{equation}
where $z,x \in \setZ$ and $n, m \in \naturals_0$. Here, the first term
on the right-hand side represents the transition in the resource level
$Z^k_t$ at the location, which is an independent Markov chain with
rates $\mu_{z,x}$. The second term on the right-hand side represents
the arrival of an agent to the location $k$ at rate
$\kappa$. The third term on the right-hand side represents the
departure of one of the $n$ agents from the location $k$. Such a
departure can only occur at a decision epoch of one of these
agents. At any such decision epoch, an agent stays with
probability $\xi(z,n)$ times the survival probability $\gamma$. Thus,
with probability $1- \gamma \xi(z,n)$, the agent leaves the location
$k$. Since there are $n$ agents at the location, each of whose
decision epoch occur at rate $\lambda$, the total rate for a departure
at the location is given by $\lambda n(1-\gamma \xi(z,n))$. Finally,
the last term on the right-hand side represents the rate of no
transition. We denote this continuous time Markov chain describing the
dynamics of a single location, where all agents adopt the Markovian
strategy $\xi$ and the rate of arrival of agents is $\kappa$, by
$\markov(\xi, \kappa)$.
 
Now, consider the decision problem faced by a single agent $i$ at
location $k$, assuming all other agents (current as well as in future)
at the location follow strategy $\xi$. 
For any fixed switching payoff $\vswitch> 0$, and arrival rate
$\kappa$, the decision problem faced by an agent $i$ can be described
as follows. As long as the agent stays at location $k$, at each
decision epoch $\tau_i^\ell$, she receives a payoff
$F(Z_i^\ell, N_i^\ell)$, and must choose whether to ``stay'' in
location $k$ or to ``switch''. Also, irrespective of this decision,
the agent's lifetime expires with probability $1-\gamma$. On choosing
to stay, with survival probability $\gamma$, the agent continues until
her next decision epoch $\tau_i^{\ell+1}$. On choosing to switch, with
survival probability $\gamma$, the agent immediately receives the
switching payoff $\vswitch$. From this description, it follows that
the decision problem facing an agent $i$ in location $k$ is an optimal
stopping problem. Denote this optimal stopping problem by
$\optstop(\xi,\kappa, \vswitch)$. In the following, we develop the
dynamic programming formulation of this problem.
 
We begin by defining the value functions for the agent. Let $V(z,n)$
denote the value function of agent $i$ at her decision epoch, prior to
her making a decision or receiving payoffs, given resource level
$z \in \setZ$ and the number of agents $n \in \naturals$ at location
$k$. Similarly, we let $\vstay(z,n)$ denote the continuation payoff of
the agent at her decision epoch, subsequent to her making the decision
to stay and conditional on her not leaving the system, given resource
level $z$ and the number of agents $n$ at location $k$. We have the
following Bellman's equation for the optimal stopping problem
$\optstop(\xi,\kappa, \vswitch)$ faced by the agent:
\begin{equation}
  \label{eq:bellman}
  \begin{split}
    V(z,n) &= F(z,n) + \gamma \max\{\vstay(z,n), \vswitch\}\\
    \vstay(z,n) &= \expec^\xi[ V(Z_\tau, N_\tau) | z, n],
\end{split}
\end{equation}
where $\expec^\xi[\cdot | z,n]$ denotes the expectation with respect
to the process defined by \eqref{eq:loc-dyn}, subject to
$(Z_0,N_0) = (z,n)$, and $\tau$ denotes the time of the first decision
epoch of the agent $i$. Here, the first equation follows from the fact
that at the decision epoch, the agent receives an immediate payoff
equal to $F(z,n)$, and has to make the decision whether to stay or
switch. Subsequent to the decision, the agent survives in the system
with probability $\gamma$. Upon choosing to switch and surviving, the
agent receives a continuation payoff equal to $\vswitch$. On the other
hand, upon choosing to stay and surviving, the agent receives a
continuation payoff equal to $\vstay(z,n)$. The second equation
relates $\vstay(z,n)$ to the expectation of the agent's value function
at the next decision epoch.

For value functions $V$ and $\vstay$ satisfying the Bellman's equation
\eqref{eq:bellman}, any optimal strategy $\xi_i$ for agent $i$ chooses
to stay if the resource level $z$ and the number of agents $n$ in the
location satisfies $\vstay(z,n) > \vswitch$, to switch if
$\vstay(z,n) < \vswitch$, and any mixed action if
$\vstay(z,n) = \vswitch$. We let $\opt(\xi,\kappa, \vswitch)$ denote
the set of all optimal strategies for the agent's decision specified
by \eqref{eq:bellman}. Specifically, for any Markovian
  strategy $\hat{\xi}$, we have
  $\hat{\xi} \in \opt(\xi, \kappa, \vswitch)$ if and only if the
  following conditions hold: $\hat{\xi}(z,n) = 1$ if
  $\vstay(z, n) > \vswitch$; $\hat{\xi}(z,n) = 0$ if
  $\vstay(z,n) < \vswitch$; and $\hat{\xi}(z,n) \in (0,1)$ only if
  $\vstay(z,n) = \vswitch$.

\subsection{Mean field equilibrium}
\label{sec:mfe-definition}
 
With the description of the model in place, we are now ready to
formally define the notion of equilibrium we focus on.
 
First, for any arrival rate $\kappa$ and the switching payoff
$\vswitch$, we require the agents play a Markov perfect equilibrium at
the location $k$. In other words, we require the strategy $\xi$ to
satisfy the following requirement: assuming all agents other than an
agent $i$ follow the strategy $\xi$, the agent $i$ maximizes her
payoff (across all possibly history-dependent strategies) by following
the strategy $\xi$. This leads us to the following condition:
\begin{equation}
\label{eq:best-response}
 \xi \in \opt(\xi, \kappa, \vswitch).
\end{equation}
 
Now, suppose for a given $\kappa$ and $\vswitch$, a Markov perfect
equilibrium $\xi$ is being played at location $k$. Then, the dynamics
of the location's state are given by $\markov(\xi, \kappa)$. Let
$\pi(\xi, \kappa)$ denote the steady state distribution of this
process. In particular, for $z \in \setZ$ and $n \geq 0$, we let
$\pi_{z,n}(\xi, \kappa)$ denote the probability that the location has
a resource level $z$ and the number of agents $n$ in steady state. (We
drop the explicit dependence of the steady state distribution on $\xi$
and $\kappa$, when the context is clear.) Thus, $\pi(\xi,\kappa)$ is
an invariant distribution under $\sQ^{\xi, \kappa}$, and satisfies
\begin{equation}
  \label{eq:steady-state}
  \sum_{z \in \setZ} \sum_{n \in \naturals_0} \pi_{z,n}(\xi,\kappa)  \sQ^{\xi, \kappa}( (z,n)
  \to (x,m)) = 0, \quad \text{for all $x \in \setZ, m \in
    \naturals_0$.}
\end{equation}
Now, consider an agent arriving to the location $k$ in steady state
$\pi(\xi, \kappa)$. We denote the total expected payoff that this
agent receives over her lifetime on following the strategy $\xi$ by
$\varrive$. Using the definition of the value function $\vstay$, we
obtain
\begin{align*}
\varrive &= \sum_{(z,n) \in \sets}  \pi_{z,n}(\xi,\kappa) \vstay(z,n+1).
\end{align*}
Here, the right hand side is obtained by observing that after the
agent arrives to the location in state $(z,n)$, which happens with
probability $\pi_{z,n}(\xi,\kappa)$, the number of agents at that location becomes
$n+1$, and the agent's continuation payoff is then $\vstay(z,n+1)$.
 
Our second condition on equilibrium requires that the total expected
payoff $\varrive$ to an agent arriving at location $k$ equals the
total expected payoff an agent at the location receives upon switching
$\vswitch$. Intuitively, we expect this condition to hold in any
symmetric equilibrium of a system with a large but finite number of
homogeneous locations, where agents  choose whether to
  stay in their current location or switch to a different location
  (chosen uniformly at random). In such a model, the switching
decisions of the agents will force the switching payoffs of all
populated locations to have the same value. Since our model of a
single location does not endogenously capture these considerations, we
impose this explicitly. In particular, we require that the switching
payoff satisfies the following equation:
\begin{equation}
 \label{eq:leave-payoff-consistency}
 \vswitch =  \sum_{(z,n) \in \sets}  \pi_{z,n}(\xi,\kappa) \vstay(z,n+1).
\end{equation}
 
The final condition we impose on the equilibrium is a requirement on
the arrival rate $\kappa$. Again, intuitively, in a symmetric
equilibrium of a large finite model with homogeneous locations, we
expect the expected number of agents at each location to be the same,
given by the agent density $\beta>0$. To capture this in our model, we
require that for a given agent density $\beta$, the arrival rate
$\kappa$ satisfies the following condition:
\begin{equation}
\label{eq:expectation-steady}
\sum_{(z,n) \in \sets}n\,\pi_{z,n} (\xi, \kappa) = \beta.
\end{equation}
 
Given these three conditions, we are now ready to define a mean-field
equilibrium:
\begin{definition}[Mean field equilibrium] A mean field equilibrium
  (MFE) consists of a strategy $\xi$, an arrival rate $\kappa$ and a
  switching payoff $\vswitch$ satisfying \eqref{eq:best-response},
  \eqref{eq:leave-payoff-consistency}, and
  \eqref{eq:expectation-steady}.
\end{definition}

Note that, in comparison to a PBE, a mean field
  equilibrium adopts a fairly natural and a vastly simpler model of
  agent behavior. In a PBE of a finite model, an agent's strategy
  depends on the state of her current location, her history, as well
  as her belief about the state of all other locations. Moreover, the
  agent constantly updates this belief based on her observations of
  the arrival process at her current location. For example, if an
  agent sees a high volume of arrivals at her current location, her
  updated belief would attribute lower resource levels at other
  locations, thereby lowering her expected payoff for switching.  Such
  complex considerations do not arise in an MFE, where the payoff from
  switching is assumed to be fixed and independent of the state
  dynamics of the current location. In a large market, this assumption
  is reasonable, as the fluctuations in the empirical distribution of
  the states of other locations are expected to cancel each other,
  analogous to a law of large numbers result\footnote{Proving this
    statement rigorously is an interesting direction for future
    work.}.  

In the next section, we show existence of a mean field equilibrium.

\section{Existence of a mean field equilibrium}
\label{sec:existence-theorem}

Below, we state the main result of the paper, proving the existence of
an MFE for general resource-sharing functions. Subsequently, in
Section~\ref{sec:threshold-struc}, we analyze the structure and
properties of a mean field equilibrium under specific assumptions on
the resource-sharing function. We have the following main theorem.
\begin{theorem}
  \label{thm:existence} For any $\lambda>0$, $\beta > 0$ and
  $\{ \mu_{z,y} > 0 : z,y \in \setZ\}$, there exists a mean field
  equilibrium $(\xi, \kappa, \vswitch)$, where $\xi(z,n) = 0$ for all
  $z \in \setZ$ and all large enough $n$.
\end{theorem}


The underlying argument behind the proof is to carefully construct a
correspondence $\map$ and show that the existence of a mean field
equilibrium is equivalent to the existence of a fixed point of
$\map$. The latter is obtained by an application of Fan-Glicksberg
fixed point theorem \citep{aliprantisB2006}. Here, we first sketch the
steps involved, and highlight the technical challenges in each of
those steps. Using these intermediate results, we then provide the
proof of Theorem~\ref{thm:existence}. (The complete proof is provided
in
Appendices~\ref{ap:existence-stationary-distn}-\ref{ap:upper-hemicontinuity}.)

\begin{enumerate}

\item We first show that for any Markovian strategy $\xi$ and arrival
  rate $\kappa>0$, the Markov chain $\markov(\xi, \kappa)$ has a
  unique invariant distribution $\pi$ satisfying
  \eqref{eq:steady-state}. This involves showing that the chain
  $\markov(\xi, \kappa)$ is irreducible and positive recurrent, which
  we accomplish by using coupling arguments to bound the chain between
  two $M/M/\infty$ queues. The proof of this result is provided in
  Appendix~\ref{ap:existence-stationary-distn}.

  Denote the (unique) invariant distribution of $\markov(\xi, \kappa)$
  by $\pi(\xi,\kappa)$. In
  Appendix~\ref{ap:continuity-stationary-distn}, by applying Berge's
  maximum theorem~\citep{berge1963topological}, we show that the
  invariant distribution $\pi(\xi, \kappa)$ is jointly continuous in
  $(\xi, \kappa)$.


\item Second, we establish that for any strategy $\xi$, there exists a
  unique value of $\kappa>0$, such that the invariant distribution
  $\pi(\xi,\kappa)$ satisfies \eqref{eq:expectation-steady}. This
  result is achieved by showing that the quantity
  $\sum_{(z,n) \in \sets} n \pi(z,n)$, where $\pi = \pi(\xi, \kappa)$
  is strictly increasing and continuous for
  $\kappa \in [\beta \lambda (1-\gamma), \beta \lambda]$ for any fixed
  $\xi$, and using the intermediate value theorem. The proof of this
  result is provided in Appendix~\ref{ap:existence-kappa}.

  Let $\kappa(\xi)$ denote the unique value of the arrival rate
  $\kappa$ for which $\pi(\xi, \kappa)$ satisfies
  \eqref{eq:expectation-steady}. The first two steps together then
  define an injective map from the strategy $\xi$ to an arrival rate
  $\kappa(\xi)$ and a steady state distribution
  $\pi(\xi, \kappa(\xi))$, such that $\pi(\xi,\kappa(\xi))$ is the
  (unique) invariant distribution of the Markov chain
  $\markov(\xi, \kappa(\xi))$, and satisfies
  \eqref{eq:expectation-steady}.

  
\item Third, we consider the decision problem
  $\optstop(\xi,\kappa(\xi), \vswitch)$ for a given strategy $\xi$ and
  switching payoff $\vswitch$. We let $\valfn(\xi, \vswitch)$ denote
  the value function satisfying the corresponding Bellman equation
  \eqref{eq:bellman}, and let $\valst(\xi, \vswitch)$ denote the
  corresponding continuation payoff function. Finally, we let
  $\valsw(\xi, \vswitch)$ denote the right-hand-side of
  \eqref{eq:leave-payoff-consistency}:
  \begin{align*}
    \mathcal{V}_{\switch}(\xi, \vswitch) = \sum_{(z,n) \in \sets} \pi_{z,n} \vstay(z,n+1),
  \end{align*}
  where $\pi = \pi(\xi,\kappa(\xi))$, and
  $\vstay = \valst(\xi, \vswitch)$.

  In Appendix~\ref{ap:uniform-bounds}, we show that these functions
  are uniformly bounded. In particular, we show that there exists
  $ 0 < \lvswitch \leq \uvswitch$, such that for all Markovian
  strategy $\xi$ and $\vswitch > 0$, we have the switching payoff
  $\valsw(\xi, \vswitch) \in [\lvswitch, \uvswitch]$. The proof of the
  uniform bounds makes extensive use of the strong Markov property for
  the chain $\markov(\xi,\kappa(\xi))$.
  

\item Fourth, we let $\mathcal{X}(\xi, \vswitch)$ denote the set of
  all optimal strategies for the agent's decision problem
  $\optstop(\xi,\kappa(\xi), \vswitch)$. Note that
  $\mathcal{X}(\xi, \vswitch) = \opt(\xi, \kappa(\xi), \vswitch)$. In
  Appendix~\ref{ap:compactness}, we identify a convex, compact set
  $\compact$ of Markovian strategies, such that if $\xi \in \compact$,
  and $\vswitch \in [\lvswitch,\uvswitch]$, then
  $\mathcal{X}(\xi, \vswitch) \subseteq \compact$. Let
  $\Upsilon = \compact \times [\lvswitch, \uvswitch]$.

  
\item Finally, we construct the correspondence
  $\map : \Upsilon \rightrightarrows \Upsilon$ defined as
  \begin{align*}
    \map(\xi, \vswitch)
    &= \mathcal{X}(\xi, \vswitch) \times \left\{ \mathcal{V}_{\switch}(\xi, \vswitch)\right\}\\
    &= \left\{ (\zeta, \mathcal{V}_{\switch}(\xi, \vswitch)) : \zeta \in \mathcal{X}(\xi, \vswitch)\right\}.
  \end{align*}
  We depict the map pictorially in Fig.~\ref{fig:map-t}. In
  Appendix~\ref{ap:upper-hemicontinuity}, we show that the
  correspondence $\map$ is upper-hemicontinuous. This requires showing
  the continuity of the value functions in $(\xi, \vswitch)$, which is
  achieved using the continuity in $\xi$
  of the process $\markov(\xi, \kappa(\xi))$ 
  under the topology of weak-convergence \citep{ethierK86}.
\end{enumerate}
We then obtain the following proof for the existence of a mean field
equilibrium.
\begin{proof}[Proof of Theorem~\ref{thm:existence}] 

  The steps outlined above show that $\map$ is an upper-hemicontinuous
  correspondence on a convex, compact subset $\Upsilon$ of a metric
  space, with values that are non-empty and convex. From an
  application of the Fan-Glicksberg fixed point theorem
  \citep{aliprantisB2006}, we obtain that $\map$ has a fixed point,
  i.e., there exists $(\xi, \vswitch) \in \Upsilon$ such that
  $(\xi, \vswitch) \in \map(\xi, \vswitch)$.
  
  Thus, by definition of $\map$, we have
  $\xi \in \mathcal{X}(\xi, \vswitch)$. This implies that $\xi$
  satisfies \eqref{eq:best-response} for the decision problem
  $\optstop(\xi, \kappa(\xi), \vswitch)$.  Second, by definition of
  $\kappa(\xi)$, we obtain that the steady state distribution
  $\pi(\xi, \kappa(\xi))$ satisfies \eqref{eq:expectation-steady}.
  Finally, from $\vswitch = \mathcal{V}_{\switch}(\xi,\vswitch)$, we
  obtain that \eqref{eq:leave-payoff-consistency} holds.  From this,
  we conclude that $(\xi, \kappa(\xi), \vswitch)$ constitutes a
  mean-field equilibrium.
\end{proof}

\begin{figure}
  \centering
  \begin{tikzpicture}[plainnode/.style={draw=none, fill=none, minimum size = 10mm}, singlearrow/.style={thick,->,shorten >=1pt, line width = 1pt, >=latex}, doublearrow/.style={thick, ->, shorten >= 1pt, line width = 1pt, >=latex, double}]
    \node[plainnode]      (originalthreshold)                        {$\xi$};
    \node[plainnode]      (kappa)             [right=1.5cm of originalthreshold] {$\kappa(\xi)$};
    \node[plainnode]      (pi)                [right=2cm of kappa] {$\pi(\xi, \kappa(\xi))$};
    \node[plainnode]      (C)                 [below=1.5cm of originalthreshold] {$\vswitch$};
    \node[plainnode]      (vstay)              [below=1.5cm of kappa] {$\mathcal{V}_\stay(\xi, \vswitch)$};
    \node[plainnode]      (ctilde)            [below=1.5cm of pi] {$\mathcal{V}_{\switch}(\xi, \vswitch) $};
    \node[plainnode]      (newthreshold)      [below=1.5cm of vstay] {$\mathcal{X}(\xi, \vswitch)$};
     
    \draw[singlearrow] (originalthreshold.east) -- (kappa.west);
    \draw (originalthreshold.north) edge[out=15, in=165, singlearrow] (pi.north west);
    \draw[singlearrow] (C.east) -- (vstay.west);
    \draw[singlearrow] (vstay.east) -- (ctilde.west);
    \draw[singlearrow] (kappa.south) -- (vstay.north);
    \draw[singlearrow] (pi.south) -- (ctilde.north);
    \draw[singlearrow] (originalthreshold.south) -- (vstay.north west);
    \draw[doublearrow] (vstay.south) -- (newthreshold.north);
    \draw[doublearrow] (C.south) -- (newthreshold.north);
  \end{tikzpicture}
  \caption{Illustration of the correspondence
    $\map(\xi, \vswitch) = \mathcal{X}(\xi, \vswitch) \times \{
    \mathcal{V}_{\switch}(\xi, \vswitch) \}$. (Here single arrows
    denote functions, and double arrows denote correspondences.)}
\label{fig:map-t}
\end{figure}
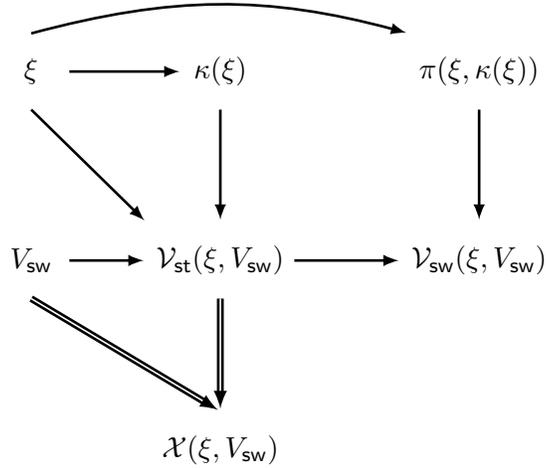

\section{Equilibrium Analysis for Decreasing Resource-Sharing
  Functions}
\label{sec:threshold-struc}

Having shown the existence of an MFE for general resource-sharing
functions, we now characterize the equilibrium strategy for the
specific case, where the resource-sharing function is non-increasing
in the number of agents at the location. Under this assumption, we
show existence of an MFE in which the equilibrium strategies have a
threshold structure. We then use this structural characterization in
Section~\ref{sec:numerics2} to compute this MFE and analyze its
welfare.

We define decreasing resource-sharing functions as follows:
\begin{definition}\label{as:monotone}
  We say that a resource-sharing function $F$ is {\em decreasing} if
  $F(z,n+1) \leq F(z,n)$ for each $z \in \setZ$ and all
  $n \in \naturals$.
\end{definition}

Decreasing resource-sharing functions appear when agents' interactions
are competitive rather than cooperative.  In
section~\ref{sec:numerics2} we consider these three examples of
decreasing resource-sharing functions.
\begin{itemize}
\item As a first example of a decreasing resource-sharing function,
  consider $F(z,n) = f(z)/n$ for some function $f$. This models
  settings where all agents at a location equally share the resource
  there. In particular, given resource level $z$ at a location, the
  $n$ agents at the location would collectively obtain total payoffs
  at rate $\lambda n F(z,n) = \lambda f(z)$, a quantity independent of
  $n$. We refer to the quantity $W(z,n) \defeq \lambda n F(z,n)$ as
    {\em single-location welfare function}.
\item Next, consider $F(z,n) = f(z)/\sqrt{n}$.  Here, the agents
  collectively receive payoffs at rate $\lambda \sqrt{n} f(z)$, which
  is increasing in $n$.  While agents compete with each other, the
  single-location welfare function increases with the number of agents
  there.
\item Finally, consider $F(z,n) = f(z)/n^{3/2}$.  This models
  extremely competitive settings, where the single-location welfare
  function decreases with the number of agents.
\end{itemize}

Before providing our result, we define threshold strategies.
Formally, for $\bx =( x_z : z \in \setZ)$, where $x_z \in \reals_+$
for each $z \in \setZ$, define the threshold strategy $\xi^\bx$ as
follows:
\begin{align*}
  \xi^\bx(z,n  ) = \begin{cases} 1 & \text{if $n < \lfloor x_z \rfloor$;}\\
    x_z - \lfloor x_z \rfloor & \text{if $n = \lfloor x_z \rfloor$;}\\
    0 & \text{otherwise.}
  \end{cases}
\end{align*}
for each $z \in \setZ$ and $n \geq 0$. In particular, under strategy
$\xi^\bx$, an agent, at her decision epoch, will stay at her current
location with resource level $z \in \setZ$ if the number of agents $n$
at the location is strictly below $\lfloor x_z \rfloor$; will switch
to a different location if $n > \lfloor x_z \rfloor$; and will stay
with probability $ x_z - \lfloor x_z \rfloor$ and switch with
remaining probability if $n = \lfloor x_z \rfloor$. 
We say that a strategy is a threshold strategy if it is of this form.

We now state our main result of this section.
\begin{theorem}
  \label{thm:threshold}
  If $F$ is a decreasing resource-sharing function, there exists an
  MFE $(\xi, \kappa, \vswitch)$ where $\xi$ is a threshold strategy.
\end{theorem}
The proof of the theorem makes essential use of the following lemma,
which states that with decreasing resource-sharing functions, the
continuation values are non-increasing.
\begin{lemma}
  \label{lem:value-decreasing}
  Let $\xi$ be a Markovian strategy, $\kappa > 0$ and $\vswitch >
  0$. If $F$ is a decreasing resource-sharing function, then for each
  $z\in \setZ$, the continuation payoff $\vstay(z,n)$ for the decision
  problem $\optstop(\xi, \kappa, \vswitch)$ is non-increasing in
  $n$. 
\end{lemma}
The proof of the lemma, provided in
Appendix~\ref{ap:proofs_threshold_strategy}, shows that the decision
problem $\optstop(\xi, \kappa, \vswitch)$ has a dynamic program that
satisfies closed convex cone properties defined in \citep{smithM2002}.
With the lemma in place, the proof of Theorem~\ref{thm:threshold}
follows from minor modifications of the argument in the proof of
Theorem~\ref{thm:existence}, and is omitted.

\section{Computation of MFE and Numerical Equilibrium Analysis}
\label{sec:numerics2}

The implications of Theorem \ref{thm:threshold} are of substantial
practical importance: when the resource-sharing function is
decreasing, the equilibrium behavior of the agents can be fully
described by $|\setZ|$ non-negative real numbers
$\{x_z : z \in \setZ\}$. This parsimony allows simple computational
methods to numerically identify an equilibrium, especially when
$|\setZ|$ is small. We use this fact to analyze the equilibrium
numerically for several representative decreasing resource-sharing
functions. We first describe our approach for computing an equilibrium
in more detail below.

\subsection{Computation of MFE}
\label{sec:computation}

To simplify notation in this section, we use 
$\bx$ to denote the threshold strategy $\xi^\bx$. Recall that an MFE is a fixed
point of the correspondence
$\map(\bx, \vswitch) = \mathcal{X}(\bx,\vswitch) \times
\valsw(\bx,\vswitch)$. For any $(\bx, \vswitch)$, we define the
distance metric $\dist$ as follows:
\begin{align*}
\label{eq:distance}
  \dist(\bx, \vswitch) = \left|\vswitch - \valsw(\bx, \vswitch)\right| + \inf_{\by \in \mathcal{X}(\bx, \vswitch)} \| \bx - \by\|_2,
\end{align*}
where $\|\cdot\|_2$ denotes the Euclidean norm. The second
term on the right-hand side denotes the distance between $\bx$ and the
set $\mathcal{X}(\bx, \vswitch)$, which is compact and convex. To find
a fixed point of $\map$, we identify a value of
$(\bx, \vswitch)$ such that $\dist(\bx, \vswitch) = 0$. We implement
two relaxations to this exact problem. First, we consider an
approximation $\dist^\epsilon$ to the metric $\dist$, obtained
primarily by truncating the state space to a finite set. Second, we
perform an adaptive search method to find a (approximate) minimizer of
the function $\dist^\epsilon$. We choose this approximate minimizer as
the value of the (approximate) MFE strategy and the corresponding
switching payoff. We describe the steps in detail below.

\begin{enumerate}
\item We truncate the state space $\sets$ of the agent's
  decision problem to $\sets_L = \setZ \times \{0, 1, \cdots, L-1\}$
  for some $L \in \naturals$. For each $\bx \in [0, L-1]^{|\setZ|}$,
  we let $\markov_L(\bx, \kappa)$ denote the Markov chain obtained by
  restricting the transitions of the chain $\markov(\bx, \kappa)$ to
  lie in the set $\sets_L$, and let $\pi_L(\bx,\kappa)$ denote its
  steady state distribution. For any
  $\bx \in [0,L]^{|\setZ|}$, the distribution $\pi_L(\bx, \kappa)$ can
  be obtained by solving a set of $L\cdot|\setZ|$ linear equations
  analogous to \eqref{eq:steady-state}.

\item For any given $\bx \in [0,L-1]^{|\setZ|}$, we perform a binary
  search over the interval $[\beta \lambda (1-\gamma), \beta \lambda]$
  to find a value $\kappa = \kappa_L(\bx)$ for which 
  \begin{align*}
    \left|\sum_{z \in \setZ} \sum_{n=0}^L n\pi_{z, n} - \beta\right| \leq \epsilon_1,
  \end{align*}
  where $\pi = \pi_L(\bx, \kappa_L(\bx))$ and $\epsilon_1 > 0$ denotes
  the tolerance level within which we seek to satisfy
  \eqref{eq:expectation-steady}.

\item For any given $\bx \in [0,L-1]^{|\setZ|}$ and
  $\vswitch \in [\lvswitch, \uvswitch]$, we then consider the decision
  problem $\optstop(\bx, \kappa_L(\bx), \vswitch)$ (with state space
  restricted to $\sets_L$). We perform value iteration to compute
  approximate value functions $\valst^\epsilon(\bx, \vswitch)$ and
  $\valsw^\epsilon(\bx, \vswitch)$, where we iterate until
  $\valst^\epsilon(\bx,\vswitch)$ is within $\epsilon_0 > 0$ (in
  sup-norm) of the limit. Using these approximate value functions, we
  identify the set of approximately optimal thresholds
  $\mathcal{X}^\epsilon(\bx, \vswitch)$. Define $\dist^\epsilon$ by
  replacing $\valsw$ and $\mathcal{X}$ in the definition of $\dist$
  with $\valsw^\epsilon$ and $\mathcal{X}^\epsilon$.

\item We seek to minimize $\dist^\epsilon(\bx, \vswitch)$ over all
  values of $\bx \in [0,L-1]^{|\setZ|}$ and
  $\vswitch \in [\lvswitch, \uvswitch]$.  We use the Nelder-Mead
  neighborhood search method \citep{nelder1965simplex} to find the
  minimizer of the distance function. To locate the global
  minimum, we run the method in parallel with multiple initial values
  of $\bx$ and $V_{\switch}$, chosen among a discretized set of
  threshold strategies
  $\markovian_L^k = \{0, (L-1)/k, 2(L-1)/k, \cdots, (k-1)(L-1)/k,
  L-1\}^{|\setZ|}$ for some $k \in \naturals$ and a discretized subset
  of $[\lvswitch, \uvswitch]$ constructed in a similar way.  
  
\item After obtaining $(\bx^*, \vswitch^*)$ that attains the minimum
  of $\dist^\epsilon$ over all runs, we do a validation check by
  comparing $\dist^\epsilon(\bx^*, \vswitch^*)$ with a threshold
  $\epsilon_2$ to see if this distance is close enough to 0 for
  $(\bx^*, \vswitch^*)$ to be an equilibrium. We accept
  $(\bx^*, \vswitch^*)$ as an approximate MFE strategy and the
  corresponding switching payoff if
  $\dist^\epsilon(\bx^*, \vswitch^*) \leq \epsilon_2$. If the
  validation check fails, a larger $k$ is chosen to provide more
  fine-grained initial starting points until a maximum number of
  iterations is reached. Although our method does not guarantee to
  find an approximate equilibrium on terminating, in all our
  computations in section~\ref{sec:statics}, we obtain an approximate
  equilibrium with corresponding $\dist^\epsilon$ smaller than
  $10^{-10}$.

\end{enumerate}

 We also note that there may be multiple equilibria in
  our model for general model parameters and resource-sharing
  functions; we have not shown uniqueness. Such instances of
  non-uniqueness may arise, for example, when the resource-sharing
  function is multimodal, as in those settings, coordination concerns
  dominate, and an agent may prefer to stay at a location if other
  agents do so, and prefer to switch if others switch.  In such
  instances, the preceding numerical procedure selects for a
  particular (approximate) equilibrium, and our comparative statics
  results in the following section correspond to the
  equilibrium\footnote{We conjecture that the equilibrium is unique
    when the resource-sharing function is decreasing and the resource
    level is binary, the setting we study for comparative statics in
    Section~\ref{sec:statics}. An extensive numerical investigation
    supports this conjecture, but we do not have a formal proof.}
  selected by this algorithm.



\subsection{Comparative statics}
\label{sec:statics}

In this section, we present the results of our numerical
investigations of the agents' behavior in a mean field equilibrium
using the computational approach described in the preceding
section. We study the setting where $\setZ = \{0 , 1\}$, with
transitions rates $\mu_{0,1} = \mu_{1,0} = \mu$. As our model is
invariant to proportional scaling of the transition rate $\mu$ and the
agents' inter-epoch rate $\lambda$, we fix $\lambda =1$. We set the
survival probability to $\gamma=0.95$. We consider decreasing
resource-sharing functions of the form $F(z,n) = z n^{-\alpha}$, where
$\alpha \in \{0.5, 1, 1.5\}$. In this setting, some locations have
resource (those with $z=1$) while others do not ($z=0$), and the
single-location welfare function is increasing for $\alpha = 0.5$,
constant for $\alpha = 1$, and decreasing for $\alpha = 1.5$ in the
number of agents there.  Finally, our approximation scheme uses
parameters $L= 200$, $k=20$, $\epsilon_0=10^{-4}$,
$\epsilon_1 =10^{-6}$ and $\epsilon_2 = 10^{-8}$.

In our computational study, we study how the model's parameters
influence both agent behavior as quantified by the equilibrium
thresholds and system efficiency as quantified by the welfare per
location.  The {\em welfare per location} is defined as the rate of
total expected payoff obtained in equilibrium by all the agents at a
location in steady state. At a location with resource level
$z \in \setZ$ and $n$ agents, the total payoff rate to those $n$
agents is given by $W(z,n) = \lambda nF(z,n)$. Since in steady state,
the state $(z,n)$ is distributed according to the mean field
distribution $\pi$, the agents' welfare per location equals
\[ W_L = \expec_\pi[ W(Z,N) ] = \sum_{z,n} \lambda nF(z,n)
  \pi_{z,n}. \] We also analyze the {\em welfare per agent}, defined
as the rate at which a randomly chosen agent receives payoff in
equilibrium. Since the agent density is equal to $\beta$, the welfare
per agent $W_A$ is given by $W_A = W_L/\beta$. When $\beta$ is held
fixed the two welfare measures are proportional, and thus we study
$W_A$ in addition to $W_L$ only when we vary $\beta$.

Figure~\ref{fig:mu_experiments} shows how the equilibrium thresholds
and the welfare per location vary as the resource process changes more
frequently, i.e, as $\mu$ increases, for a fixed value of
$\beta = 20$.
For each resource-sharing function, for small values of $\mu$, the
difference between the thresholds $x_1$ and $x_0$ is
substantial. Since the resource level changes slowly, an agent in a
location with resource is willing to suffer significant competition
(in the form of other agents) before choosing to switch her
location. Note that, as $\alpha$ increases, the level of competition
at which agents switch decreases, consistent with our observation that
as $\alpha$ increases, competition becomes more severe. On the other
hand, as $\mu$ increases, the difference in the two thresholds
diminishes. This is because increasing $\mu$ diminishes the benefit of
staying in a location.  As the resource levels change more frequently,
the resource process mixes more readily and thus future resource
levels are less correlated with current levels.

Figure~\ref{fig:mu_experiments} also shows that the welfare per
location depends crucially on the resource-sharing function.  When the
single-location welfare function increases with the number of agents
at that location ($\alpha = 1/2$), the welfare per location decreases
as resource levels change more frequently, i.e., as $\mu$
increases. In contrast, when the single-location welfare function
decreases with the number of agents there ($\alpha = 3/2$), the
welfare per location increases as $\mu$ increases.  To understand
this, observe that when $\mu$ is small, the thresholds $x_1$ and $x_0$
are well-separated, implying that the agents will be concentrated in
locations with positive resource level. On the other hand, when $\mu$
is large, the two thresholds are similar, and agents are more
equitably distributed between locations with and without
resource. When $\alpha < 1$, the former distribution of the agents
obtains more welfare per location, since single-location welfare
function is increasing with the number of agents at a location with
resource, and having more agents at these locations increases welfare.
On the other hand, when $\alpha > 1$, the former distribution incurs
lower welfare per location due to severe competition among the agents
at the location with resource. (When $\alpha=1$, the distribution of
the agents between locations with or without resource does not
substantially affect the welfare per location. In particular, as long
as a location with resource has at least one agent present, the total
payoff at that location is the same.)

\begin{figure}
  \centering
  \includegraphics[width=0.9\linewidth]{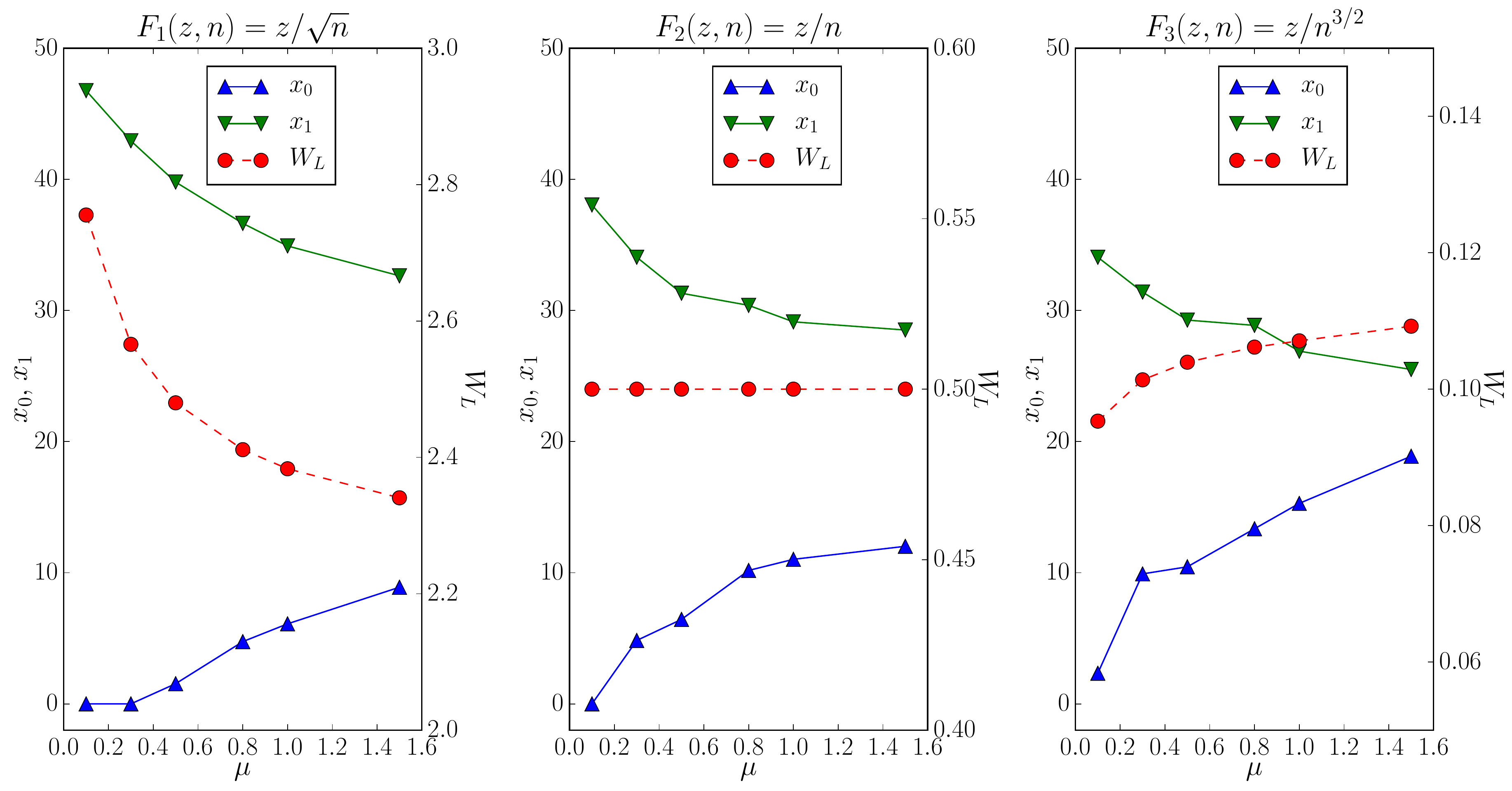}
  \caption{Equilibrium thresholds and  welfare under different resource transition rates $\mu$, with agent density fixed at $\beta=20$.  }
  \label{fig:mu_experiments}
\end{figure}

Figure~\ref{fig:beta_experiments_ag} shows equilibrium properties as a
function of the agent density $\beta$ when resource levels change
slowly ($\mu=0.25$).  The difference between the thresholds $x_1$ and
$x_0$ widens as $\beta$ increases for each resource-sharing function.
This is because increasing $\beta$ for any fixed state $(z,n)$ at the
current location diminishes an agent's expected payoff from switching,
since there are more agents to compete against. Thus, when the current
location has resource, the agents become more likely to stay as
$\beta$ gets larger.

We further observe that, as $\beta$ increases, the welfare per
location increases when $\alpha = 0.5$, decreases when $\alpha = 1.5$,
and is essentially constant when $\alpha =1$. As in
Figure~\ref{fig:mu_experiments}, this relation is explained by the
equilibrium distribution of agents between locations with and without
resource, arising from the dependence of the equilibrium thresholds on
$\beta$: as the difference between the two thresholds increases, the
welfare per location increases when $\alpha =0.5$, and decreases when
$\alpha = 1.5$. However, since the degree of competition increases as
$\beta$ increases, we observe that irrespective of the
resource-sharing function the welfare per agent decreases. 

 The preceding comparative statics reveals an important
  feature of our dynamic model and its equilibrium that is lacking in
  a static analysis: our analysis captures the joint distribution of
  the agents and the resource levels across
  locations. Figure~\ref{fig:mu_experiments} demonstrates this by
  showing that agents' strategies change as the resource transition
  rate $\mu$ changes. In contrast, since all values of $\mu$ result in
  the same steady-state proportion ($50\%$) of locations in each
  resource state, a static analysis that only tracks the stationary
  resource state distribution would generate the same market outcomes
  for all values of $\mu$. Furthermore, the welfare also changes with
  $\mu$ for resource-sharing functions other than $z/n$, where the
  total payoff rate in a location $\lambda nF(z,n)$ depends
  non-trivially on $n$. Such an effect would not materialize in a
  static model which ignores the dynamics of the resource process and
  tracks only the steady state.


\begin{figure}
  \centering
  \includegraphics[width=0.9\linewidth]{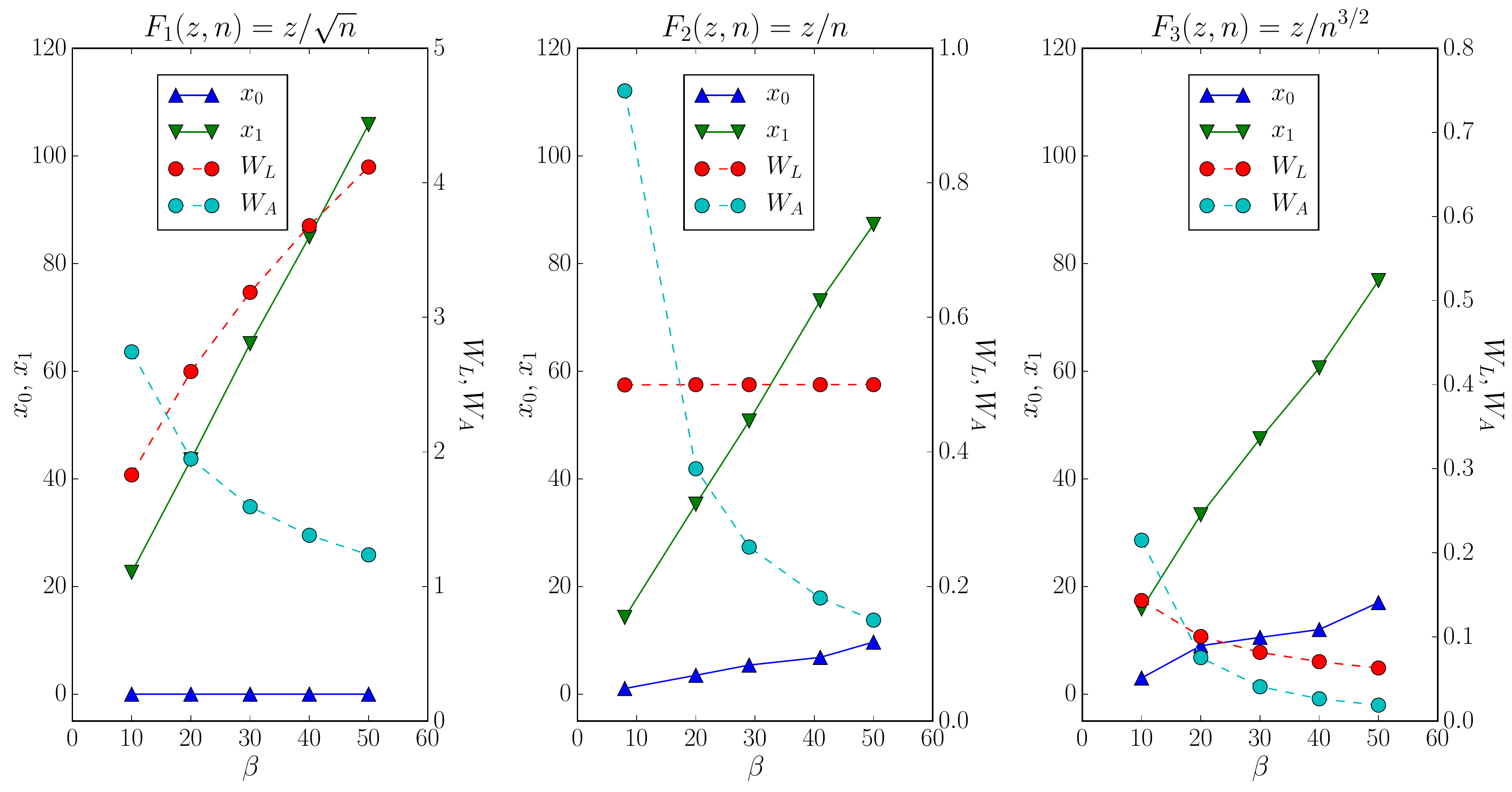}
  \caption{Equilibrium thresholds and welfare under different agent
    densities $\beta$. $W_A$ is multiplied by 15 for all values. Note
    the resource transition rate is given by $\mu=0.25$.  }
     \label{fig:beta_experiments_ag}
\end{figure}

  \subsection{Case study: Setting platform commission}
\label{sec:application}

In this section, we provide a case study to illustrate how our model
can be used to evaluate engineering interventions. Specifically, we
apply our model to the ride-hailing market in Manhattan. Ride-hailing
platforms charge a commission when they transfer rider payments to
their driver partners, and consequently, the drivers' behavior in the
market is influenced by this commission rate. In this case study, we
investigate how different commission rates affect the aggregate
revenue of the drivers and the platform (and how it is split between
the two); the outcome of this analysis provides a reference for
platforms when an adjustment of commission rate is under
consideration.

We view taxi drivers as agents, different neighborhoods of Manhattan
as locations, and taxi trip demand as the resource in our model. We
assume the drivers, at the end of each day, decide for the next day
whether to stay in the same neighborhood or switch to another one. We
also assume a driver makes this decision based on the trip demand in
his current neighborhood as well as his estimate of the number of
competing drivers in the same neighborhood.

Below we describe how the model parameters are estimated, and further
describe the assumptions. We use the yellow cab trip records from the
New York City Taxi and Limousine Commission dataset~\citep{NYCtrip} to
estimate these parameters. The data limitations prevent us from
performing a full-blown analysis; in such instances, we use our
judgment to assign parameter values. We set the parameter values as
follows:
\begin{itemize}
\item Agent density $\beta$: We divide Manhattan into 12 regions, with
  the diameter of each region approximately equal to the average taxi
  trip length in Manhattan. The agent density is then estimated as
  $\beta = 400$ drivers per location, following an estimate of 4800
  active taxi drivers, obtained by averaging across different times of
  day.
	
\item Resource process $\{\mu_{z, y}\}$: We assume a resource model
  with binary states, with $0$ denoting the typical resource state,
  and $1$ denoting a high resource state. Such a high resource may
  describe local conditions (such as local events, weather patterns,
  etc) that temporarily lead to high demand for rides. To estimate the
  transition rates between the two states, we use weather as a proxy,
  and estimate the transition between rainy and non-rainy days using
  historical weather data from Manhattan \citep{Weather}. This yields
  a transition rate of $\mu_{0,1} = 1/3.86$ and $\mu_{1, 0} = 1/1.93$,
  with units $\text{day}^{-1}$. These values are a reasonable proxy
  for state transitions, indicating a high resource state
  approximately every $4$ days, for a duration of about $2$
  consecutive days.

 	
\item Payment function $F$: Most ride-hailing platforms
    use dynamic pricing mechanisms to improve market efficiency, and
    such mechanisms can be designed to increase the aggregate revenue
    with the number of drivers \citep{CKW17, Chen2016}, as increased
    driver availability allows more trips to happen. However, at the
    same time, higher competition among the drivers decreases the
    revenue received by an individual driver. To model these aspects,
    we let the aggregate revenue rate from riders at a location with
    resource state $z$ and $n$ drivers equal $n^{1-\alpha} f(z)$ for
    some parameter $\alpha \in (0,1)$, where $f(z)$ captures the
    dependence on the resource state. This entails the revenue rate
    per driver to equal $f(z)/n^{\alpha}$ and hence the rate of
    payment to an individual driver in the location takes the
    following form:
  \[ F(z, n) = \frac{(1-c(z)) f(z)}{n^\alpha},\quad z\in\{0,1\}, n \in
    \mathbb{N},\] where $c(z)$ denotes the (resource-dependent)
  commission rate charged by the platform. For our analysis, we choose
  $\alpha = 0.5$.
  To estimate $f(0)$, we use the
  average daily rider payment on non-rainy days in Manhattan from
  \citep{NYCtrip}, which yields an estimate of $1.2 \times 10^4$
  dollars per hour per location.  We do not, however, estimate $f(1)$
  using rider payments on rainy days, since our data comes from yellow
  cab data with fixed prices, whereas modern ride-hailing platforms
  typically increase price as demand increases. We therefore assume
  the average total rider payment when the resource is high $(z=1)$ to
  be $20\%$ higher, and set $f(1) = 1.2f(0)$.

\item Decision rate: We choose $\lambda = 1 \text{ day}^{-1}$.
 	
\item Survival probability: We choose $\gamma = 0.995$, indicating a
  planning horizon of $1/\lambda (1-\gamma) = 200$ days.
\end{itemize}
 
Assuming a baseline commission rate of $15\%$ in both resource states,
we investigate how the revenue of drivers and the platform would vary
under a number of commission rate scenarios.  For each such
combination of $c(0)$ and $c(1)$ (and under parameter values described
above), we numerically compute the resulting mean field equilibrium in
our model, and the driver and platform revenues in the computed
equilibrium. We share these results in Table~\ref{table}. These
results can be used to access the magnitude of the impact, and to
decide whether commission should be raised in aggregate, or if it
would be better to selectively raise it based on demand (resource
states). A table such as this could be shared with decision makers as
part of a larger decision process.
\begin{table}[tb]
\centering
\begin{tabular}{|c|c|c|c|c|c|c|c|}
  \hline
  $c(0)$  & $c(1)$ & $\mathsf{DriRev}$  & $\Delta \mathsf{DriRev}$  & $\mathsf{PlatRev}$ & $\Delta \mathsf{PlatRev}$ & $\mathsf{AggRev}$ & $\Delta \mathsf{AggRev}$\\
  \hline
  0.15 & 0.15 &  26.121 & - & 4.610 & -  & 30.731 & - \\
  \hline
  0.175& 0.175 & 25.353 & -2.94\% & 5.378 & 16.66\% & 30.731 & 0.00\% \\
  \hline
  0.15 & 0.20 & 25.504 & -2.36\% & 5.219 & 13.20\% & 30.723 & -0.02\% \\
  \hline
  0.20 & 0.15 & 25.210 & -3.49\% & 5.507 & 19.46\% & 30.718 & -0.04\%\\
  \hline
  0.2& 0.2 & 24.584 & -5.88\% & 6.146 & 33.32\% & 30.730 & 0.00\%\\
  \hline
\end{tabular}
\caption{Here, $c(z)$ is the commission rate at resource state
  $z \in \{0,1\}$. $\mathsf{DriRev}$, $\mathsf{PlatRev}$ and
  $\mathsf{AggRev}$ denote the revenue (in units $10^5$ dollars per
  hour) for drivers, for the platform, and in aggregate,
  respectively. $\Delta \mathsf{DriRev}$ denotes the change in
  drivers' revenue compared with the base case ($c(0) = c(1) = 0.15$),
  with $\Delta \mathsf{PlatRev}$ and $\Delta \mathsf{AggRev}$ defined
  similarly. \label{table}}
\end{table}

As discussed earlier, our dynamic model allows us to
  capture the joint distribution of the drivers and the aggregate
  revenue across locations. The distribution of the drivers across
  locations is important because it influences a driver's payoff upon
  switching, which influences her switching decisions. Our model
  enables us to include this endogenous effect of the driver
  distribution on the drivers' switching decisions in evaluating
  different commission rates. Without the dynamics (and the tractable
  equilibrium concept of a mean field equilibrium), such effects would
  be hard to incorporate in a static analysis, rendering it
  incomplete.

\section{Conclusion}
\label{sec:conclusion}

Our results establish that in equilibrium, the agents in our
  model base their decision to explore solely on the state of the
  location they currently reside in, and on its steady state
  distribution. In particular, our results justify analyzing
spatio-temporal models under simple yet optimal models of agent
behavior.

Our model and analysis raise many topics for future
research. First, we have used the notion of a mean field
  equilibrium to analyze a single location in isolation, assuming the
  other locations are described by the mean field distribution. A
natural question is whether the resulting strategy constitutes an
approximate equilibrium in the system with large but finite number of
agents and locations. Such approximation results for mean field
equilibrium have been obtained in other contexts (see, for example,
\citep{sachin_2010,balseiro2014,iyerJS2014}). In the finite system, a
single agent visits multiple locations over her lifetime, inducing
correlations among the states of those locations.  The
  analytical challenge in obtaining an approximation result involves
  showing that as the system size increases, such correlations vanish,
  and in the limit, the dynamics of a location in the finite system
  approaches the dynamics of the single-location in our model.

  On the modeling front, we have assumed that each location is
  homogeneous. In particular, we assume the resource process is
  distributed independently and identically across different
  locations. One consequence of this homogeneity is
    agents do not choose their destination when they switch.  It is
    straightforward to extend our model and the analysis to
    incorporate location heterogeneity and to let agents choose
    their destination when switching. Such a model would better
    represent the settings we study.  For example, in ride-hailing
    settings, residential neighborhoods have different demand
    characteristics than business districts, and drivers choose the
    neighborhood to operate in based on these characteristics. A
    formal finite model of this extension has multiple {\em types} of
    locations, where each location has type-dependent resource
    dynamics, resource-sharing function and agent density. Agents
    choose not only whether to stay or switch, but also which location
    type to switch to, whereupon the destination is chosen uniformly
    among locations of that type. Using similar arguments for the
    homogeneous setting, we can obtain the corresponding limiting
    infinite system meant to capture the limiting behavior as the
    number of agents and the locations increases with the type
    distribution and the agent density fixed.  Our proof of existence
    of a mean field equilibrium applies in this setting with minor
    modifications; we omit the details due to space considerations.
  On the other hand, we have assumed that the resource process at a
  location is exogenous specified, whereas an extension could allow
  for the resource transitions at a location to depend on the number
  of agents therein.

Finally, our work also sets the stage for analyzing engineering
interventions and their economic impact. One such intervention
involves altering the resource sharing function at each location
through subsidies or penalties to induce the agents to stay in or
switch from a location, thereby affecting their welfare. A further
question is whether sharing information about locations' states would
benefit or harm the agents, and how such an information sharing
mechanism should be designed.  Answers to these questions would help
platforms such as Uber or Airbnb to increase their efficiency.

\appendix

\section{Description of the finite system}
\label{ap:finite-model}

In this section, we provide a formal description of the system with
finite number of locations and agents (the ``finite system''), which
motivates our mean field model.

The finite system has a set of locations $\calK$, where each location
$k \in \calK$ contains a stochastic time-varying resource. We use
$Z_t^k$ to denote the resource level at location $k$ at time
$t \geq 0$. We assume the resource process $\{Z_t^k : t \geq 0\}$ is a
finite state continuous time Markov chain, and further assume the
resource processes across different locations in the system are
distributed identically and independently. We let $\setZ$ denote the
set of values the resource process can take, and let $\mu_{zy}>0$
denote the transition rate of $Z_t^k$ from a state $z \in \setZ$ to a
state $y \in \setZ$. Furthermore, we make the assumption that each
process $Z_t^k$ is irreducible and positive recurrent, with a unique
invariant distribution given by
$\{\pi_{\resource}(z) : z \in \setZ\}$.

Spread across this set of locations are $N$ agents. Each agent may
switch between locations in search for resources and less competition,
as we detail below. Each agent $i$ is associated with a Poisson clock
with rate $\lambda$, such that each time the clock rings, the agent
decides whether to stay in the location or switch to another one. We
refer to each clock ring of agent $i$ as the agent's decision epoch,
and let $\tau_i^\ell$ and $k_i^\ell$ denote the time and location of
her $\ell^{th}$ decision epoch respectively.

We let $N_t^k$ denote the number of agents at the location $k$ at time
$t$.  At each decision epoch $t=\tau_i^\ell$, the agent $i$ at
location $k=k_i^\ell$ receives a payoff $F(Z^k_t, N^k_t)$ that depends
on the resource level $Z^k_t$ and the number of agents $N^k_t$ at that
location. We make the same assumptions on $F$ as in
Section~\ref{sec:formal_model}.

Subsequent to receiving the payoff, the agent $i$ makes the decision
whether to continue staying at her location or move to a different
location. On choosing to move to a different location, agent $i$
instantaneously arrives at a new location $k_i^{\ell+1}$. We make the
assumption that the new location $k_i^{\ell+1}$ is drawn independently
and uniformly from the set of all locations other than the agent's
current location. Note that this assumption precludes us from modeling
an agent's strategic choice of \emph{which} location to move
to. Nevertheless, we make this assumption as, even under this
restrictive assumption, the analysis of the agent's decision problem
turns out to be challenging. In Section~\ref{sec:conclusion}, we
discussed a few extensions and modifications that align closer to
practical settings.

Similar to the mean field model, we assume agents in the finite system
are short-lived: after each decision epoch $\tau_i^\ell$, subsequent
to making her decision regarding whether to stay in her current
location or move to a different location, the agent $i$ departs the
system independently with probability $1-\gamma$, never to return, and
we denote as $\tau_i$ the time she leaves the system. We also assume
for each agent that departs, a new agent arrives to the system at a
location chosen uniformly at random, to maintain constant system size,
same as in the mean field model.

Finally, we describe the utility and the information structure of each
agent in the model. We assume that each agent $i$, at each time $t$,
at her current location $k$, observes the resource level $Z^k_t$ and
the number of agents $N^k_t$. On the other hand, the agent cannot
observe the resource level and the number of agents at any other
location.  We assume the agents have perfect recall, and hence, at any
decision epoch $\tau_i^\ell$, agent $i$ bases her decision to stay or
move on the entire history (namely the resource levels and the number
of agents at each location she has visited) she has observed until
that time.

Given this informational assumption, each agent $i$ is risk-neutral
and wants to maximize the total expected payoff accrued over her
lifetime. Formally, each agent $i$ seeks to maximize
\begin{align*}
  \expec\left[ \sum_{\ell = 1}^\infty F(Z_i^\ell, N_i^\ell) \ind\{ \tau_i^\ell \leq \tau_i \} \right],
\end{align*}
where the expectation is over the randomness in the resource levels,
the arrival and departure process of the agents, and their (and their
competitors') strategies. Since the departure of an agent is
independent of the rest of the system, it is straightforward to show
that the agent's expected payoff can be equivalently written as
\begin{align*}
  \expec\left[ \sum_{\ell = 1}^\infty \gamma^{\ell-1} F(Z_i^\ell, N_i^\ell) \right].
\end{align*}
Thus, each agent $i$'s decision problem is equivalent to the decision
problem faced by a persistent agent (who never departs the system)
seeking to maximizer her total expected discounted payoff.

\section{Existence and uniqueness of invariant distribution of
  $\markov(\xi, \kappa)$}
\label{ap:existence-stationary-distn}
In this section, we show that for any Markovian strategy $\xi$ and
arrival rate $\kappa>0$, the Markov chain $\markov(\xi, \kappa)$ has a
unique steady state distribution.
\begin{lemma}
  \label{lem:ergodicity}
  For any Markovian strategy $\xi$ and arrival rate $\kappa \geq 0$,
  there exists a unique steady state distribution for
  $\markov(\xi, \kappa)$ satisfying \eqref{eq:steady-state}.
\end{lemma}

\begin{proof} Fix a Markovian strategy $\xi$, and an arrival rate
  $\kappa>0$.  We prove the lemma statement by showing that the Markov
  chain $\markov(\xi,\kappa)$ is irreducible and positive
  recurrent. The fact that $\markov(\xi, \kappa)$ is irreducible
  follows straightforwardly from \eqref{eq:loc-dyn} and the fact that
  the resource process is independent and ergodic. Thus, it only
  remains to show that the chain is positive recurrent.

  Let $\sets_0 \defeq \{ (z,0) : z \in \setZ\}$ and define
  $T_z(\sets_0)$ is the first return time of the chain to $\sets_0$,
  given it starts at $(z,0)$:
  \begin{align*}
   T_z(\sets_0) \defeq \inf\{ & t > 0 : (Z_t,N_t) \in \sets_0,
    (Z_s,N_s) \notin \sets_0 \text{ for some $0 < s < t$, } \\ 
   &  \text{given } (Z_0,N_0) = (z,0)\}.
  \end{align*}
  In the following, we show that $T_z(\sets_0)$ has finite expectation
  for each $z \in \setZ$. From this, using the ergodicity of the
  resource process, it follows that the return time to a particular
  state $(z_0,0) \in \sets_0$ also has finite expectation, and hence
  the chain is ergodic.

  To show that $T_z(\sets_0)$ has finite expectation, we use a
  coupling argument. Given a Markov chain
  $(Z_t,N_t) \sim \markov(\xi,\kappa)$ with $(Z_0,N_0) = (z,0)$, we
  construct a coupled process
  $\nagent{t}{1} \sim \mminfty{(1-\gamma)\lambda}{\kappa}$ with
  $\nagent{t}{1} = 0$, as in the proof of
  Lemma~\ref{lem:coupling}. Define $\tilde{T}_z(0)$
  to be the first return time to $0$ of the chain $\nagent{t}{1}$.
  From the construction of the coupling, it follows that
  $N_t \leq \nagent{t}{1}$ for all $t\geq 0$, and hence
  $T_z(\sets_0) \leq \tilde{T}_z(0)$. Thus, we have
  $\expec[T_z(\sets_0)] \leq \expec[\tilde{T}_z(0)]$. The result then
  follows immediately from the fact that an
  $\mminfty{(1-\gamma)\lambda}{\kappa}$ queue is ergodic, and hence
  $\expec[\tilde{T}_z(0)] < \infty$ for all $z \in \setZ$.
\end{proof}

\section{Joint continuity of the invariant distribution of
  $\markov(\xi, \kappa)$}
\label{ap:continuity-stationary-distn}

In the following, we show that the steady state distribution
$\pi(\xi, \kappa)$ of the Markov chain $\markov(\xi, \kappa)$ is
jointly (and uniformly) continuous in its parameters.  This continuity
result will play an important role in subsequent results that
constitute our proof of existence of an MFE.

To prove the continuity of $\pi(\xi, \kappa)$, we adopt an approach
similar to \citep{levanS07}, where we characterize the invariant
distribution of $\markov(\xi, \kappa)$ as a maximizer of a continuous
function, and apply Berge's maximum theorem.  Before we present the
formal argument, we specify the topologies (and the metric) we impose
on the set of Markovian strategies and the set of invariant
probability distributions, and specify the continuous function
$\Lambda$ that we consider.  First, we endow the state space
$\sets = \setZ \times \naturals_0$ with the discrete topology.  Let
$\bounded$ denote the set of bounded function $h : \sets \to
\reals$. (Note that since we impose the discrete topology on $\sets$,
any such $h$ is also continuous.) We endow $\bounded$ with the
sup-norm:
\begin{align}
\label{eq:d-define}
\|h_1 - h_2\|_\infty \defeq \sup_{(z,n) \in \sets} \left|h_1(z,n) - h_2(z,n) \right|,  \quad \text{for $h_1, h_2 \in \bounded$}
\end{align}
Let $\markovian \subseteq
\bounded$ denote the set of Markovian strategies, with the topology
induced from $\bounded$.

We let $\signed$ denote the set of finite signed measures on
$\sets$, and we endow
$\signed$ with the weak topology, which is equivalent to the
  topology induced by $\ell_1$-norm since $\sets$ is countable:
\begin{align}
\label{eq:rho-define}
  \| \mu - \nu \|_1 = \sum_{(z,n) \in \sets} | \mu(z,n) - \nu(z,n) |, \quad \text{for $\mu, \nu \in \signed$.} 
\end{align}
Let $\Gamma = \{ \pi(\xi, \kappa) : \xi \in \Pi, \kappa \in [\beta
\lambda (1- \gamma), \beta \lambda] \} \subseteq
\signed$ denote the set of invariant distributions (with the induced
topology) for all Markovian strategies and arrival rates. Let
$\overline{\Gamma}$ denote the closure of $\Gamma$.

For $\xi \in \markovian$,
$\kappa \in [\beta \lambda (1-\gamma), \beta\lambda]$ and
$\nu \in \overline{\Gamma}$, define $\Lambda(\xi, \kappa, \nu)$ as
follows:
\begin{align}
\label{eq:f-define-1}
  \Lambda(\xi, \kappa, \nu) \defeq - \sum_{(z,n) \in \sets} \frac{1}{n+1} \left| (\nu \sQ)(z,n) \right|, 
\end{align}  
where $\sQ = \sQ^{\xi, \kappa}$ denotes the transition kernel of
$\markov(\xi, \kappa)$, and $\nu \sQ$ is defined as
\begin{align*}
  (\nu \sQ) (z, n)  =  \sum_{(y,m) \in \sets} \nu(y,m) \sQ( (y,m) \to (z,n)).
\end{align*}

With the preliminaries in place, we are now ready to state the main
lemma of this section.
\begin{lemma}
  \label{lem:stationary-dist-continuity}
  The map $(\xi, \kappa) \mapsto \pi(\xi, \kappa)$ is jointly (and
  uniformly) continuous in $(\xi,\kappa)$ for $\xi \in \markovian$ and
  $\kappa \in [\beta \lambda (1-\gamma), \beta \lambda]$.
\end{lemma}
\begin{proof}
  In Lemma~\ref{lem:tightness}, we show that the set of distributions
  $\Gamma$ is uniformly tight. Then, from Prohorov's theorem
  \citep{billingsley2013}, we obtain that $\overline{\Gamma}$ is
  compact. Observe that
  \begin{align}
    \label{eq:maximizer}
    \underset{\nu \in \overline{\Gamma}}{\arg\max} ~\Lambda(\xi, \kappa, \nu) = \{\pi(\xi, \kappa)\}.
  \end{align}
  This follows from the fact that $\pi(\xi, \kappa)$ is the unique
  probability distribution over $\sets$ for which
  \eqref{eq:steady-state} holds.

  In Lemma~\ref{lem:kernel-continuity}, we show that
  $\Lambda(\xi, \kappa, \nu)$ is jointly (and uniformly) continuous
  its parameters for $\xi \in \markovian$, $\nu \in \overline{\Gamma}$
  and $\kappa \in [\beta \lambda (1-\gamma) , \beta \lambda]$. The
  result then follows from a direct application of Berge's maximum
  theorem \citep{berge1963topological} to \eqref{eq:maximizer}.
\end{proof}

The following two auxiliary lemmas are used in the proof of
Lemma~\ref{lem:stationary-dist-continuity}.
\begin{lemma}
    \label{lem:tightness}
    The set $\Gamma$ of invariant distributions is tight.
  \end{lemma}
  \begin{proof}

    We prove the lemma using a coupling argument. For any
    $\xi \in \markovian$ and
    $\kappa \in [\beta \lambda (1 - \gamma) , \beta \lambda]$, let
    $(Z_t, N_t) \sim \markov(\xi, \kappa)$, with $(Z_0, N_0) = (z, n)$
    for some $(z, n) \in \sets$. Let $\pi$ denote the invariant
    distribution of $\markov(\xi, \kappa)$. Independently, let
    $\widehat{N}_t \sim \mminfty{\lambda(1-\gamma)}{\beta\lambda}$
    with $\widehat{N}_0 = n$. Let $\widehat{\pi}$ denote the invariant
    distribution of $\mminfty{\lambda(1-\gamma)}{\beta\lambda}$; it is
    straightforward to show that $\widehat{\pi}$ is Poisson with mean
    $\beta/(1-\gamma)$.

    Using Lemma~\ref{lem:compare-mminfty-queue-arrival} and
    Lemma~\ref{lem:coupling}, we obtain that $N_t$
    is (first-order) stochastically dominated by $\widehat{N}_t$ for
    all $t \geq 0$. From this, we obtain (by taking limits and using
    ergodicity) that for all $k >0$, we have
    \begin{align*}
      \sum_{z \in \setZ} \sum_{n > k} \pi(z,n) \leq \sum_{n > k} \widehat{\pi}(n).
    \end{align*}
    For any $\epsilon > 0$, choose a $k^\epsilon > 0$ such that
    $\sum_{n > k^\epsilon} \widehat{\pi}(n) < \epsilon$. (Such a
    $k^\epsilon$ exists, given that $\widehat{\pi}$ is Poisson with
    finite mean.)  This implies that
    \begin{align*}
      \sum_{z \in \setZ} \sum_{n > k^\epsilon} \pi(z,n) < \epsilon, \quad \text{for all $\epsilon > 0$.}
    \end{align*}
    Since $k^\epsilon$ is independent of the choice of
    $(\xi, \kappa)$, we obtain that $\Gamma$ is tight.
\end{proof}
The following lemma proves the joint continuity of $\Lambda$.
\begin{lemma}
    \label{lem:kernel-continuity}
    The function $\Lambda$ as defined in \eqref{eq:f-define-1} is
    jointly (and uniformly) continuous.
\end{lemma}
\begin{proof}
  Consider $\xi_i \in \markovian$,
  $\kappa_i \in [\beta\lambda(1-\gamma), \beta\lambda]$ and
  $\nu_i \in \overline{\Gamma}$ for $i=1,2$. We let $\sQ_i$ denote the
  transition kernel of $\markov(\xi_i, \kappa_i)$. We have
  \begin{align}
    \label{eq:triangle-split}
  &  | \Lambda(\xi_1, \kappa_1, \nu_1) - \Lambda(\xi_2, \kappa_2, \nu_2) | \notag \\
  &\leq | \Lambda(\xi_1, \kappa_1, \nu_1) - \Lambda(\xi_2, \kappa_2, \nu_1) |  + | \Lambda(\xi_2, \kappa_2, \nu_1) - \Lambda(\xi_2, \kappa_2, \nu_2) |.
\end{align}
Now, note that
\begin{align}
  \label{eq:first-term}
    & |\Lambda(\xi_1, \kappa_1, \nu_1) - \Lambda(\xi_2, \kappa_2, \nu_1) |\notag\\
  \leq & \sum_{(z,n) \in \sets} \frac{1}{n+1} \left| (\nu_1 \sQ_1)(z,n)  - (\nu_1 \sQ_2)(z,n)\right|\notag\\
  \leq & \sum_{(z,n) \in \sets} \sum_{(y,m) \in \sets} \frac{1}{n+1}  \nu_1(y,m) \left|  \sQ_1 ( (y,m) \to  (z,n) )  -  \sQ_2( (y,m) \to (z,n))\right|\notag\\
  \leq & \sum_{(y,m) \in \sets}   \nu_1(y,m) \sum_{(z,n) \in \sets}  \frac{1}{n+1} \left|  \sQ_1 ( (y,m) \to  (z,n) )  -  \sQ_2( (y,m) \to (z,n))\right|\notag\\
  \leq & \sup_{(y,m) \in \sets} \sum_{(z,n) \in \sets}  \frac{1}{n+1} \left|  \sQ_1 ( (y,m) \to  (z,n) )  -  \sQ_2( (y,m) \to (z,n))\right|.
\end{align}
Now, using \eqref{eq:loc-dyn}, we obtain that
\begin{align*}
  & \sQ_1((y,m) \to (z,n)) -  \sQ_2((y,m) \to (z,n)) \\
  = & \ind\{z=y, n=m+1\} (\kappa_1 - \kappa_2) + \ind\{z=y, n=m-1\} \lambda \gamma m (\xi_1(y,m) - \xi_2(y,m))\\
  & - \ind\{z=y,n=m\} \left( \kappa_1 - \kappa_2 + \lambda m \gamma (\xi_1(y,m) - \xi_2(y,m))\right),
\end{align*}
and hence
\begin{align*}
  &\sum_{(z,n) \in \sets}  \frac{1}{n+1} \left|  \sQ_1 ( (y,m) \to  (z,n) )  -  \sQ_2( (y,m) \to (z,n))\right|\\
  \leq  &\left(\frac{1}{m+2} + \frac{1}{m+1}\right) |\kappa_1 - \kappa_2| + \lambda \gamma \left(1+ \frac{m}{m+1}\right) | \xi_1(y,m) - \xi_2(y,m)|\\
  \leq & 2 \left(|\kappa_1 - \kappa_2| + \lambda \gamma  | \xi_1(y,m) - \xi_2(y,m)|\right).
\end{align*}
Thus, from \eqref{eq:first-term}, we obtain
\begin{align}
  \label{eq:bound-first-term}
 |\Lambda(\xi_1, \kappa_1, \nu_1) - \Lambda(\xi_2, \kappa_2, \nu_1) |  &\leq 2 |\kappa_1 - \kappa_2| + 2\lambda \gamma  \| \xi_1 - \xi_2\|_\infty .
\end{align}
Next, observe that
\begin{align}
  \label{eq:second-term}
  & | \Lambda(\xi_2, \kappa_2, \nu_1) - \Lambda(\xi_2, \kappa_2, \nu_2) | \notag \\
  \leq & \sum_{(z,n) \in \sets} \frac{1}{n+1} | (\nu_1\sQ_2) (z,n) - (\nu_2\sQ_2) (z,n)| \notag\\
  \leq & \sum_{(z,n) \in \sets} \frac{1}{n+1} \sum_{(y,m) \in \sets}\left|\sQ_2((y,m) \to (z,n))\right| | \nu_1(y,m) - \nu_2(y,m) | \notag\\
  \leq  & \sum_{(y,m) \in \sets} | \nu_1(y,m) - \nu_2(y,m) | \sum_{(z,n) \in \sets} \frac{1}{n+1} \left|\sQ_2((y,m) \to (z,n))\right|.
\end{align}
Now, again from \eqref{eq:loc-dyn} and after some straightforward
algebra, we obtain that
\begin{align*}
  &\sum_{(z,n) \in \sets} \frac{1}{n+1} \left|\sQ_2((y,m) \to (z,n))\right| \\
  \leq & \frac{2}{m+1}\sum_{z \neq y} \mu_{y,z} + \kappa_2 \left(\frac{1}{m+2} + \frac{1}{m+1}\right) + \lambda (1 - \gamma \xi_2(y,m)) \left(1 + \frac{m}{m+1}\right)\\
  \leq & \sum_{z \neq y} \mu_{y,z} + 2 \kappa_2 + 2 \lambda \\
  \leq & \max_{y \in \setZ} \sum_{z \neq y} \mu_{y,z} + 2 (\beta +1) \lambda,
\end{align*}
where we have used the fact that $\kappa_2 \leq \beta \lambda$ in the
last inequality. Thus, we from \eqref{eq:second-term}, we obtain
\begin{align}
  \label{eq:bound-second-term}
  | \Lambda(\xi_2, \kappa_2, \nu_1) - \Lambda(\xi_2, \kappa_2, \nu_2) | &\leq   \left(\max_{y \in \setZ} \sum_{z \neq y} \mu_{y,z} + 2 (\beta +1) \lambda \right) \| \nu_1 - \nu_2\|_1.
\end{align}
Therefore, combining \eqref{eq:triangle-split},
\eqref{eq:bound-first-term} and \eqref{eq:bound-second-term}, we
obtain
\begin{align*}
  &  | \Lambda(\xi_1, \kappa_1, \nu_1) - \Lambda(\xi_2, \kappa_2, \nu_2) | \\
\leq &  2 |\kappa_1 - \kappa_2| + 2\lambda \gamma  \| \xi_1 - \xi_2\|_\infty  +   \left(\max_{y \in \setZ} \sum_{z \neq y} \mu_{y,z} + 2 (\beta +1) \lambda \right) \| \nu_1 - \nu_2\|_1.
\end{align*}
Thus, $\Lambda$ is Lipschitz, and hence jointly and (uniformly)
continuous in its parameters.
\end{proof}

\section{Existence of $\kappa$ satisfying equilibrium condition}
\label{ap:existence-kappa}

In this section we show for any Markovian strategy $\xi$, there exists
a unique arrival rate
$\kappa \in [\beta\lambda(1-\gamma), \beta\lambda]$ for which the
steady state distribution $\pi$ of the Markov chain
$\markov(\xi, \kappa)$ satisfies the equation
\eqref{eq:expectation-steady}.

Towards that goal, for any Markovian strategy $\xi$ and arrival rate
$\kappa>0$, define
\[\phi(\xi, \kappa) \defeq \sum_{(z, n) \in \sets} n \pi^{\xi,\kappa}(z,n) \]
where $\pi^{\xi,\kappa}$ is the unique steady state distribution of
$\markov(\xi, \kappa)$. We seek to show that there exists a
$\kappa \in [\beta\lambda(1-\gamma), \beta\lambda]$ such that
$\phi(\xi, \kappa) = \beta$. We prove this result using intermediate
value theorem. First, we show that $\phi(\xi, \kappa)$ is a strictly
increasing function of $\kappa$ for any given $\xi \in
\markovian$. Second, we show
$\phi(\xi, \beta\lambda(1-\beta)) \leq \beta$ and
$\phi(\xi, \beta\lambda) \geq \beta$, which implies any $\kappa$ such
that $\phi(\xi,\kappa)=\beta$ must lie in
$[\beta\lambda(1-\gamma), \beta\lambda]$. The result then follows once
we show $\phi(\xi, \kappa)$ is a continuous function of $\kappa$.

In the rest of this section, we assume that the strategy $\xi$ is
fixed, and drop the explicit dependence on $\xi$ from notation
wherever convenient.  We now proceed with the first-step.

\subsection{Strict monotonicity of $\phi(\cdot)$}
\label{ap:kappa-increasing}

\begin{lemma}
    \label{lem:expec-agents-increasing}
    Given any Markovian strategy $\xi$, $\phi(\kappa)$ is a strictly
    increasing function of $\kappa$ on
    $[\beta\lambda(1-\gamma), \beta\lambda]$.
\end{lemma}

\begin{proof}
  For any
  $\kappa_1, \kappa_2 \in [\beta\lambda(1-\gamma), \beta\lambda]$ with
  $\kappa_1 < \kappa_2$, consider two coupled chains
  $\lstate{t}{i} \sim \markov(\xi, \kappa_i)$ for $i=1,2$, as in the
  proof of Lemma~\ref{lem:coupling}, where
  $\zstate{t}{1} = \zstate{t}{2}$ and
  $\nagent{t}{1} \leq \nagent{t}{2}$ for all $t \geq 0$.
  For $i=1,2$, we have
  \[ \frac{1}{t} \int_0^t \nagent{s}{i} ds \rightarrow \sum_{z, n}
    n\pi_i(z, n) = \phi(\kappa_i) \] almost surely as
  $t \rightarrow \infty$, where we write $\pi_i$ for
  $\pi^{\kappa_i}$. Since $\nagent{t}{1} \leq \nagent{t}{2}$ for all
  $t$, we have $\phi(\kappa_1) \leq \phi(\kappa_2)$.

  Next, suppose for the sake of contradiction that
  $\phi(\kappa_1) = \phi(\kappa_2)$. Since
  $\zstate{t}{1} = \zstate{t}{2}$ and
  $\nagent{t}{1} \leq \nagent{t}{2}$ for all $t \geq 0$, we have
\begin{align}
\label{eqn:dominance-ind}
 \ind\{\zstate{t}{1}=z, \nagent{t}{1} \geq n \} \leq \ind\{\zstate{t}{2}=z, \nagent{t}{2} \geq n \}, 
\end{align}
for all $(z,n) \in \sets$, $t \geq 0$.

For any $(z,n) \in \sets$, we have
\begin{align}
\label{eqn:slln-ind}
 \frac{1}{t} \int_0^t \ind\{\zstate{s}{i}=z, \nagent{s}{i} \geq n \}ds  \rightarrow \sum_{n' \geq n} \pi_i(z,n')
\end{align}
almost surely as $t \rightarrow \infty$, for $i=1,2$. By
\eqref{eqn:dominance-ind} and \eqref{eqn:slln-ind} we have
\[ \sum_{n' \geq n} \pi_1(z, n') \leq \sum_{n' \geq n} \pi_2(z, n'), \]
for any $(z,n)$, and
\[ \phi(\kappa_1) = \sum_{n \geq 0} n \left( \sum_{z \in \setZ}
    \pi_1(z, n) \right) = \sum_{n \geq 0} \sum_{z \in \setZ, n' > n}
  \pi_1(z, n') \leq \sum_{n \geq 0} \sum_{z \in \setZ, n' > n}
  \pi_2(z, n') = \phi(\kappa_2). \] Since by our assumption
$\phi(\kappa_1) = \phi(\kappa_2)$, the inequality in the preceding
equation is actually an equality. This implies
\[ \sum_{n' \geq n} \pi_1(z, n') = \sum_{n' \geq n} \pi_2(z, n'), \]
for all $(z,n)$, which further implies that $\pi_1$ and $\pi_2$ are
the same distribution.

For $i=1,2$, the equation \eqref{eq:steady-state} implies
\begin{align*}
  & \pi_i(z,n)\left(\kappa_i + \sum_{y \ne z} \mu_{z,y} + \lambda n(1-\gamma\xi(z,n))\right)  \\
  = & \pi_i(z, n-1)\kappa_i + \sum_{y \ne z}\mu_{y, z}\pi_i(y,n) + \pi_i(z, n+1)\lambda(n+1)(1 - \gamma\xi(z, n+1)),
\end{align*}  
which leads to
\begin{align}
\label{eq:balance-eqn}
  & \kappa_i(\pi_i(z,n) - \pi_i(z, n-1)) \notag \\
  = &  \sum_{y \ne z}\mu_{y, z}\pi_i(y,n) + \pi_i(z, n+1)\lambda(n+1)(1 - \gamma\xi(z, n+1))  \notag \\
  & - \pi_i(z,n) \left(\sum_{y \ne z} \mu_{z,y} + \lambda n(1-\gamma\xi(z,n))\right).
\end{align}  
Since $\pi_1 = \pi_2$, the right hand side of \eqref{eq:balance-eqn}
is the same for $i=1,2$, hence we have
\[ \kappa_1(\pi_1(z,n) - \pi_1(z,n-1)) = \kappa_2(\pi_2(z,n) - \pi_2(z,n-1)). \]
But $\kappa_1 < \kappa_2$ and $\pi_1 = \pi_2$ implies that for $i=1,2$, $\pi_i(z,n) = \pi_i(z,n-1)$ for all $(z,n)$, hence $\pi_i$ cannot be a probability distribution over $\sets$, and this contradiction completes the proof.
\end{proof}

\subsection{Bounds for $\phi(\cdot)$}

In this section, we provide bounds on the function $\phi(\kappa)$ for
any $\kappa>0$. These bounds immediately imply that for
$\kappa = \beta \lambda$, $\phi(\kappa) \geq \beta$, and for
$\kappa = \beta \lambda (1- \gamma)$, $\phi(\kappa) \leq
\beta$. Together with Lemma~\ref{lem:expec-agents-increasing}, this
implies that any $\kappa$ for which $\phi(\kappa) = \beta$ must lie in
the interval $[\beta \lambda (1- \gamma), \beta \lambda]$.

\begin{lemma}
   \label{lem:bound-expec-agents}
   For any Markovian strategy $\xi$ and arriving rate $\kappa \geq 0$,
   $\phi(\kappa)$ satisfies
   \[\frac{\kappa}{\lambda} \leq \phi(\kappa) \leq
     \frac{\kappa}{\lambda(1-\gamma)}.\]
 \end{lemma}

 \begin{proof}
  	Let $(Z_t, N_t) \sim \markov(\xi, \kappa)$, and $(Z_0, N_0) = (z,n)$. Denote as $\mminfty{\lambda}{\kappa}$ an (independent) $M/M/\infty$ queue with arrival rate $\kappa$ and service rate $\lambda$, and let $\nagent{t}{i} \sim \mminfty{\lambda_i}{\kappa}$ for $i = 1, 2$ be two independent processes with $\nagent{0}{1} = \nagent{0}{2} = n$,
  where $\lambda_1 = \lambda$ and $\lambda_2 =
  (1-\gamma)\lambda$. 
   	
  Let $\pi^\kappa$ be the steady state distribution of
  $\markov(\xi, \kappa)$, and $\pi_i$ be the steady state distribution
  of $\mminfty{\lambda_i}{\kappa}$ for $i=1,2$. We have
  \begin{align*}
    \frac{1}{t}\int_0^t N_s ds \rightarrow  \sum_{z,n} n\pi^\kappa(z,n),
  \end{align*}
  and 
  \begin{align*}
    \frac{1}{t}\int_0^t \nagent{s}{i} ds \rightarrow \sum_{n} n\pi_i(n), \quad i=1,2
  \end{align*}
  almost surely as $t\rightarrow \infty$. From Lemma
  \ref{lem:coupling}, we have
  $\nagent{t}{1} \sdleq N_t \sdleq \nagent{t}{2}$ for all $t\geq 0$,
  therefore we have
  \[ \sum_{n} n\pi_1(n) \leq \sum_{z,n} n\pi^\kappa(z,n) \leq
    \sum_{n}n \pi_2(n).\] The result then follows from the fact that
  for $i=1,2$, $\pi_i$ is Poisson distribution with mean
  $\kappa/\lambda_i$.
\end{proof}

\subsection{Continuity of $\phi(\cdot)$}
\label{ap:parametric-continuity}

Observe that the existence of a
$\kappa \in [\beta \lambda (1-\gamma) , \beta \lambda]$ such that
$\phi(\kappa) = \beta$ would follow immediately once we prove the
continuity of $\phi(\cdot)$ in $\kappa$ for any fixed Markovian
strategy $\xi$. In this section, we prove a stronger statement, namely
that $\phi(\xi, \kappa)$ is jointly continuous in $(\xi, \kappa)$.


\begin{lemma}
  \label{lem:expec-agents-continuity}
  The map $\phi(\xi, \kappa)$ is jointly and uniformly continuous in
  $(\xi, \kappa)$ for Markovian $\xi$ and for
  $\kappa \in [\beta\lambda(1-\gamma), \beta\lambda]$.
\end{lemma}
\begin{proof}
  Given Markovian strategies $\xi_1$ and $\xi_2$, and arriving rates
  $\kappa_1$, $\kappa_2 \in [\beta\lambda(1-\gamma), \beta\lambda]$ ,
  let $\pi_i$ be the steady state distribution of
  $\markov(\xi_i, \kappa_i)$, for $i=1,2$. We have, for any arbitrary
  $k \geq 0$,
  \begin{equation}
    \label{eq:bound-phi-difference}
    \begin{split}
      &\left| \phi(\xi_1, \kappa_1) - \phi(\xi_2, \kappa_2) \right|\\
       = & \left| \sum_{z, n} n(\pi_1(z,n) - \pi_2(z,n))\right| \\
      \leq & \left|\sum_{z\in\setZ} \sum_{ n \leq k} n(\pi_1(z,n) -\pi_2(z,n) )\right| + \left|\sum_{z\in\setZ} \sum_{ n > k} n(\pi_1(z,n) -\pi_2(z,n) )\right| \\
       \leq & \left|\sum_{z\in\setZ}\sum_{ n \leq k} n(\pi_1(z,n)
        -\pi_2(z,n) )\right| + \sum_{z\in\setZ} \sum_{ n >k}
      n\pi_1(z,n) + \sum_{z\in\setZ, n > k} n \pi_2(z,n).
  \end{split}
\end{equation}
Now, bounding the first term, we obtain
\begin{align}
\label{eq:bound-first-terms}
 &  \left|\sum_{z \in \setZ}\sum_{ n \leq k} n(\pi_1(z,n) -\pi_2(z,n))\right|  \notag\\
  \leq & k \sum_{z \in \setZ}\sum_{ n \leq k} \left| \pi_1(z,n) -\pi_2(z,n) \right| \leq k \| \pi_1 - \pi_2\|_1.
\end{align}
To bound the other terms, we use a coupling argument. Let
$ (\zstate{t}{i}, \nagent{t}{i}) \sim \markov(\xi_i, \kappa_i)$ with
$(\zstate{0}{i}, \nagent{0}{i}) = (z, n)$ for $i=1,2$. Let
$\widehat{N}_t \sim \mminfty{\lambda(1-\gamma)}{\beta\lambda}$, with
$\widehat{N}_0 = n$, denote the number of agents in an (independent)
$M/M/\infty$ queue at time with arrival rate $\beta\lambda$ and
service rate $\lambda(1-\gamma)$.  Let $\widehat{\pi}$ denote the
steady state distribution of $\widehat{N}_t$. By
Lemma~\ref{lem:compare-mminfty-queue-arrival} and
Lemma~\ref{lem:coupling}, we have
$\nagent{t}{i} \sdleq \widehat{N}_t$ for all $t \geq 0$ and for each
$i=1,2$. From this stochastic dominance, it is straightforward to
obtain that
\begin{align}
  \label{eq:bound-tail-prob}
  \sum_{z \in \setZ, n > k} n \pi_i(z,n) \leq \sum_{n > k} n \widehat{\pi}(n), \quad i=1,2.
\end{align}

Thus, from\eqref{eq:bound-phi-difference}, \eqref{eq:bound-tail-prob}
and \eqref{eq:bound-first-terms}, we have
\begin{equation*}
|\phi(\xi_1, \kappa_1) - \phi(\xi_2, \kappa_2)| \leq k \|\pi_1 -
\pi_2\|_1 + 2\sum_{n > k} n \widehat{\pi}(n).
\end{equation*}
Now, for any $\epsilon > 0$, choose $k$ such that
$\sum_{n > k} n \widehat{\pi}(n) < \epsilon/4$. (Note that this choice
of $k$ is independent of $(\xi_i, \kappa_i)$ and depends only on the
steady state $\widehat{\pi}$ of
$\mminfty{\lambda(1-\gamma)}{\beta \lambda}$, which is Poisson with
mean $\beta/(1-\gamma)$.) Second, from
Lemma~\ref{lem:stationary-dist-continuity}, we obtain that for any
$\epsilon>0$, there exists a $\delta>0$ such that for all
$(\xi_1, \kappa_1)$ and $(\xi_2, \kappa_2)$ such that
$\|\xi_1 - \xi_2\|_\infty < \delta$ and $|\kappa_1-\kappa_2| < \delta$, we have
$\| \pi_1 - \pi_2\|_1 < \epsilon/2k$. Taken together, we obtain that
for any $\epsilon>0$, there exists a $\delta>0$ such that for all
$(\xi_1, \kappa_1)$ and $(\xi_2, \kappa_2)$ such that
$\|\xi_1- \xi_2\| < \delta$ and $|\kappa_1-\kappa_2| < \delta$, we have
$|\phi(\xi_1, \kappa_1) - \phi(\xi_2, \kappa_2)| < \epsilon$. Thus, we
obtain that $\phi(\cdot)$ is jointly and uniformly continuous.
\end{proof}

\subsection{Continuity}

For any $\xi \in \markovian$, let $\kappa(\xi)$ denote the unique
value of $\kappa$ for which $\pi(\xi, \kappa)$ satisfies
\eqref{eq:expectation-steady}. Below, we show that $\kappa(\xi)$ is a
continuous function of $\xi$.

\begin{lemma}
    \label{lem:continuity-kappa}
    The map $\xi \mapsto \kappa(\xi)$ is continuous.
\end{lemma}

\begin{proof}
  Define $W(\xi, \kappa) = - | \beta - \phi(\xi, \kappa)|$. Note that,
  from Lemma~\ref{lem:expec-agents-increasing}, we obtain
  \begin{align*}
    \underset{\kappa \in [\beta \lambda (1-\gamma), \beta \lambda]}{\arg\max} W(\xi, \kappa) =  \{ \kappa(\xi)\}.
  \end{align*}
  From Lemma~\ref{lem:expec-agents-continuity}, we obtain that
  $\phi(\xi, \kappa)$ is jointly continuous in $(\xi, \kappa)$, and
  hence so is $W(\xi, \kappa)$. The result then follows from Berge's
  maximum theorem \citep{berge1963topological}.
\end{proof}

\section{Uniform bounds on value functions}
\label{ap:uniform-bounds}

For a given $\xi \in \markovian$ and $\vswitch > 0$, we seek to study
the decision problem $\optstop(\xi, \kappa(\xi), \vswitch)$. Before we
proceed, we need some definitions. Let
$(Z_t, N_t) \sim \markov(\xi, \kappa(\xi))$, and let
$\expec^\xi(\cdot | z,n)$ denote the expectation-operator with respect
to $\{ (Z_t, N_t) : t \geq 0\}$ conditioned on $(Z_0, N_0) =
(z,n)$. Fix an agent, say agent $1$, among all the agents at time $0$,
and let $\tau$ be the agent's first decision epoch.

Let
$\dynamic: \markovian \times \reals_+ \times \bounded \to \bounded$
denote the Bellman-operator for the agent's decision problem
$\optstop(\xi, \kappa(\xi), \vswitch)$, where for any
$\xi \in \markovian$, $\vswitch > 0$ and $U \in \bounded$, the
function $W= \dynamic(\xi, \vswitch, U)$ is defined as follows:
\begin{equation}
  \label{eq:dp-operator}
  W(z, n) = F(z,n) + \gamma \max\left\{\expec^\xi\left[ U(Z_\tau, N_\tau) | z,n\right], \vswitch\right\}, \quad \text{for all $(z,n) \in \sets$.}
\end{equation}

The following lemma states that the map
$\dynamic(\xi, \vswitch, \cdot)$ is a contraction. The proof follows
from standard arguments and is omitted.
\begin{lemma}
  \label{lem:bellman-contraction}
  For any $\xi \in \markovian$ and $\vswitch > 0$, we have
  $\dynamic(\xi, \vswitch,U) \in \bounded$ for all $U \in
  \bounded$. Furthermore, the map
  $\dynamic(\xi, \vswitch, \cdot) : \bounded \to \bounded$ is a
  contraction (with contraction parameter $\gamma$) for any
  $\xi \in \markovian$ and $\vswitch >0$.
\end{lemma}
Let $\valfn(\xi, \vswitch) \in \bounded$ be the unique fixed
point of $\dynamic(\xi, \vswitch, \cdot)$. Define
$\valst(\xi, \vswitch) \in \bounded$ and
$\valsw(\xi, \vswitch) \in \reals_+$ as follows:
\begin{equation}
  \label{eq:stswdef}
  \begin{split}
  \valst( z,n; \xi, \vswitch) &= \expec^\xi\left[ \valfn(Z_\tau, N_\tau ;\xi, \vswitch)  | z,n\right]\\
  \valsw(\xi, \vswitch) &= \sum_{(z,n) \in \sets} \pi(z,n) \valst(z,n+1; \xi, \vswitch),
\end{split}
\end{equation}
where $\pi = \pi(\xi, \kappa(\xi))$. Here $\valfn(z,n;\xi, \vswitch)$
(and $\valst(z,n;\xi, \vswitch)$) denote the value taken by
$\valfn(\xi, \vswitch)$ (resp., $\valst(\xi, \vswitch)$) at
$(z,n) \in \sets$.


We begin this section by providing bounds on $\valsw$, $\valst$ and
$\valfn$. Define
\begin{align*}
\uvswitch &= \frac{1}{1-\gamma}\|F\|_\infty,\\
\lvswitch &= \exp\left(-\frac{\beta}{1-\gamma}\right)  \sum_{(z,n) \in
  \sets}\frac{\beta^n(1-\gamma)^n}{(1+\beta+\Psi)^{n+1} (n+1)!}
\pi_\resource(z) F(z,n+1) > 0,
\end{align*}
where
$\Psi = \frac{1}{\lambda}\max_{z \in \setZ} \sum_{y \neq z} \mu_{zy}
\in (0,\infty)$, and $\pi_\resource$ is the steady state distribution
of the resource process.
The following lemma, providing a uniform upper bound on the value
functions, follows immediately from definition.
\begin{lemma}\label{lem:upperbound-val}
  For any $\xi \in \markovian$ and $\vswitch>0$, the value functions
  satisfy
  $|\valsw(\xi, \vswitch)| \leq \|\valst(\xi, \vswitch)\|_\infty \leq
  \|\valfn(\xi, \vswitch)\|_\infty \leq \uvswitch = \frac{\|
    F\|_\infty}{1-\gamma}$.
\end{lemma}
\begin{proof}
  Observe that from \eqref{eq:stswdef}, we have
  $|\valsw(\xi, \switch)| \leq \| \valst(\xi, \switch)\|_\infty \leq \|
  \valfn(\xi, \switch)\|_\infty$. Also, from the fact that
  $\valfn(\xi, \vswitch)$ is the fixed-point of
  $\dynamic(\xi, \vswitch, \cdot)$, we obtain
  \begin{align*}
    \| \valfn(\xi, \vswitch)\|_\infty &\leq \|F\|_\infty + \gamma \max\{\|\valfn(\xi, \vswitch)\|_\infty, |\valsw(\xi, \switch)|\} = \|F\|_\infty + \gamma \|\valfn(\xi, \vswitch)\|_\infty.
  \end{align*} 
  Rearranging, we obtain that
  $\|\valfn(\xi, \vswitch)\|_\infty \leq \frac{1}{1-\gamma}
  \|F\|_\infty = \uvswitch$.
\end{proof}
The next lemma provides a uniform lower bound on the value
functions. The proof makes extensive use of the strong Markovian
property for the chain $\markov(\xi,\kappa(\xi))$.
\begin{lemma}\label{lem:lowerbound-val}
  For any $\xi \in \markovian$ and $\vswitch>0$, we have
  $\valsw(\xi, \vswitch) \geq \lvswitch$.
\end{lemma}
\begin{proof}
  Observe that
  \begin{align*}
    \valfn(z,n; \xi, \vswitch)
    &= F(z,n) + \gamma \max\{ \valst(z,n; \xi, \vswitch) ,  \vswitch\} \geq F(z,n).
  \end{align*}
  Recalling the definition of $\valst(\xi, \vswitch)$ and using the
  (strong) Markov property, we obtain
\begin{align*}
  \valst(z,n;\xi, \vswitch) 
  = & \frac{\lambda}{n\lambda + \kappa(\xi) + \sum_{y \neq z} \mu_{zy}  } \valfn(z,n;\xi, \vswitch) \\
  & + \frac{\kappa(\xi)}{n\lambda + \kappa(\xi) + \sum_{y \neq z} \mu_{zy}  } \valst(z,n+1;\xi, \vswitch)\\
  & + \sum_{w \neq z} \frac{\mu_{wz}}{n\lambda + \kappa(\xi) + \sum_{y \neq z} \mu_{zy}  } \valst(w,n;\xi, \vswitch) \\
  &  + \frac{(n-1)\lambda(1 - \gamma \xi(z,n))}{n\lambda + \kappa(\xi) + \sum_{y \neq z} \mu_{zy}  } \valst(z,n-1;\xi, \vswitch)\\
  & + \frac{(n-1)\lambda \gamma \xi(z,n)}{n\lambda + \kappa(\xi) + \sum_{y \neq z} \mu_{zy}  } \valst(z,n;\xi, \vswitch)\\
  \geq  & \frac{\lambda}{n\lambda + \kappa(\xi)  + \sum_{y \neq z} \mu_{z y}} \valfn(z,n;\xi, \vswitch)\\
  \geq & \frac{\lambda}{\lambda (n+\beta)  + \sum_{y \neq z} \mu_{z y}} F(z,n),
\end{align*}
where the last line follows from the fact that
$\kappa(\xi) \leq \beta \lambda$. Using the definition of $\Psi$, we
obtain
\begin{align}
\label{eq:lowerbound-valst}
  \valst(z,n;\xi, \vswitch) \geq \frac{1}{n+\beta + \Psi} F(z,n) \geq
  \frac{1}{n(1+\beta + \Psi)}F(z,n). 
\end{align}

Next, observe that $\pi = \pi(\xi, \kappa(\xi))$ satisfies the
steady-state equation \eqref{eq:steady-state}:
\begin{align*}
  \sum_{(y,m) \in \sets} \pi(y,m) \sQ^\xi((y,m) \to (z,n)) = 0,
\end{align*}
where $\sQ^{\xi}$ denote the transition kernel of the Markov chain
$\markov(\xi, \kappa(\xi))$. Using the expression \eqref{eq:loc-dyn}
for $\sQ^\xi$, we obtain
\begin{align*}
  &\pi(z,n) (\kappa(\xi) + \sum_{y \neq z} \mu_{zy} + \lambda n (1 - \gamma \xi(z,n)))\\
  &\quad =  \pi(z,n-1)\kappa(\xi) + \sum_{w \neq z} \pi(w,n) \mu_{wz}  + \pi(z,n+1)\lambda (n+1)(1-\gamma \xi(z,n+1)).
\end{align*}
This implies that 
\begin{align*}
  \pi(z,n) 
  &\geq \pi(z,n-1) \frac{\kappa(\xi)}{\kappa(\xi) + \sum_{y \neq z} \mu_{zy} + \lambda n (1 - \gamma \xi(z,n))}\\
  &\geq \pi(z,n-1) \frac{\beta\lambda (1-\gamma) }{\beta \lambda(1-\gamma)  + \sum_{y \neq z} \mu_{zy}  + \lambda n }\\
  &\geq \pi(z,n-1) \frac{\beta (1-\gamma) }{\beta (1-\gamma)  + \Psi  +  n }\\
  &\geq \pi(z,n-1) \frac{\beta (1-\gamma) }{(1 + \beta+ \Psi)n }.
\end{align*}
Thus, we obtain
\begin{align*}
  \pi(z,n) \geq \pi(z,0) \frac{\beta^n(1-\gamma)^n}{(1+\beta+\Psi)^n n!}.
\end{align*}
Now, from Lemma~\ref{lem:coupling}, we obtain that
the process $(Z_t,N_t) \sim \markov(\xi, \kappa(\xi))$ with
$(Z_0, N_0) = (z,n)$ is stochastically dominated by $(Z_t,N^{(1)}_t)$
where $N^{(1)}_t$ is an (independent)
$\mminfty{\lambda(1-\gamma)}{\beta \lambda}$ process with
$N^{(1)}_0 = n$. Hence, we have
$\pi(z,0) \geq \pi_\resource(z)\prob(N^{(1)}_\infty = 0)$, where the
steady state $N^{(1)}_\infty$ is given by a Poisson distribution with
parameter $\beta/(1-\gamma)$, implying
$\prob(N^{(1)}_\infty = 0) = \exp(-\beta/(1-\gamma))$. Thus, we obtain
\begin{align}
  \label{eq:lowerbound-pi}
  \pi(z,n) &\geq \pi_\resource(z) \exp\left(-\frac{\beta}{1-\gamma}\right) \frac{\beta^n(1-\gamma)^n}{(1+\beta+\Psi)^n n!}.   
\end{align}
Finally, from \eqref{eq:stswdef}, we have
  \begin{align*}
    \valsw(\xi, \vswitch) 
    &= \sum_{(z,n) \in \sets} \pi(z,n) \valst(z, n+1; \xi, \vswitch)\\
    &\geq \sum_{(z,n) \in \sets} \pi(z,n)\frac{F(z,n+1)}{(1+\beta+\Psi)(n+1)}\\
    &\geq \exp\left(-\frac{\beta}{1-\gamma}\right)\sum_{(z,n) \in \sets}\frac{\beta^n(1-\gamma)^n}{(1+\beta+\Psi)^{n+1} (n+1)!} \pi_\resource(z)  F(z,n+1)\\
    &= \lvswitch.
  \end{align*}
  where we use \eqref{eq:lowerbound-valst} in the first inequality and
  \eqref{eq:lowerbound-pi} in the second.
\end{proof}

\section{A compact set of Markovian strategies}
\label{ap:compactness}
For $\xi \in \markovian$ and $\vswitch \in [\lvswitch, \uvswitch]$,
denote the set of all optimal Markovian strategies for the decision
problem $\optstop(\xi, \kappa(\xi), \vswitch)$ by
$\mathcal{X}(\xi, \vswitch) \subseteq \markovian$. In particular,
$\mathcal{X}(\xi, \vswitch)$ is the set of all $\zeta \in \markovian$
such that
\begin{align*}
  \zeta(z,n) = \begin{cases} 1 & \text{ if $\valst(z,n; \xi, \vswitch) > \vswitch$;}\\
    0 & \text{ if $\valst(z,n; \xi, \vswitch) <  \vswitch$.}
  \end{cases}
\end{align*}
It is straightforward to show that the set
$\mathcal{X}(\xi, \vswitch)$ is non-empty and convex.

In this section, we provide characterization of a compact set
$\compact \subseteq \markovian$ of strategies such that if
$\xi \in \compact$ and $\vswitch \in [\lvswitch, \uvswitch]$, then
$\mathcal{X}(\xi, \vswitch) \subseteq \compact$. This characterization
is later used to define a correspondence over a compact set to which
we apply the Kakutani fixed point theorem to show the existence of an
MFE. (Note that the set $\markovian$ is not compact under the
sup-norm.)



We begin by defining the set $\compact$. Recall that $F(z,n) \to 0$ as
$n \to \infty$ for all $z \in \setZ$. Let $K_0$ be defined as
\begin{align*}
  K_0 &= \inf\left\{ m : F(z,n) < \frac{(1-\gamma)^2}{2}\lvswitch \text{ for all $z \in \setZ$ and $n \geq m$} \right\},
\end{align*} and let $K_1$ be defined as
\begin{align*}
  K_1 = \inf\left\{ n : \exp\left(-\frac{1}{8}\sqrt{n-1}\right) + \frac{2}{\sqrt{\log(n-1)}} + \gamma^{\lfloor\sqrt{\log(n-1)}\rfloor}(1-\gamma) < \frac{(1-\gamma)^2\lvswitch}{4\|F\|_\infty } \right \}
\end{align*}
Let $K_{\max} = \max\{4 K_0^2 +1, K_1\}$. Define the set
$\compact \in \markovian$ as follows:
\begin{align*}
  \compact = \{ \xi \in \markovian : \xi(z,n) = 0 \text{ for all $z \in \setZ$ and $n \geq K_{\max}$}\}.
\end{align*}
In other words, under any strategy $\xi \in \compact$, each agent
chooses to switch her location, if the number of agents at her
location is greater than $K_{\max}$, irrespective of the resource
level. It is straightforward to show that $\compact$ is compact, by
noting that it is isomorphic to $[0,1]^{K_{\max}}$ under the Euclidean
topology.

The following lemma states that if $\xi \in \compact$ and
$\vswitch \geq \lvswitch$, then the optimal action for an agent at the
state $(z,n)$ is to switch if $n \geq K_{\max}$.
\begin{lemma}
  \label{lem:convergence-vstay}
  For $\xi \in \compact$ and $\vswitch \geq \lvswitch$, we have
  $\valst(z,n; \xi, \vswitch) < \vswitch$ for all $z \in \setZ$ and
  for all $n \geq K_{\max}$.
\end{lemma}
\begin{proof}
  Consider an agent $i$ in location $k$ facing the decision problem
  $\optstop(\xi, \kappa(\xi), \vswitch)$ for a given
  $\xi \in \compact$ and $\vswitch \geq \lvswitch$.  Let $\tau^\ell > 0$
  denote the time of the $\ell^{th}$-decision epoch of the agent, for
  $\ell = 1,2, \cdots$. Let $(Z_t,N_t)$ denote the state of the
  location at time $t$, and for brevity, we let $(Z_\ell, N_\ell)$
  denote $(Z_{\tau^\ell}, N_{\tau^\ell})$ for each
  $\ell = 1,2, \cdots$.

  Suppose $(Z_0,N_0) = (z,n)$ for $z \in \setZ$ and $n \geq
  K_{\max}$. Fix a strategy $\phi \in \markovian$ for the agent, and
  let $\tau^\phi$ denote the first time at which the agent chooses to
  switch under $\phi$. Let $\vstay^\phi(z,n)$ denote agent $i$'s
  continuation payoffs under the strategy $\phi$, subsequent to her
  making the decision to stay and not leaving the system, given the
  state of the location $(z,n)$. We have the following expression for
  the $\vstay^\phi(z,n)$:
  \begin{align}
    \label{eq:total-reward}
    \vstay^\phi(z,n)
    &= \expec\left[\sum_{\ell = 1}^\infty \gamma^{\ell-1} F(Z_\ell,N_\ell) \ind\{ \tau^\ell \leq \tau^\phi\} + \gamma^\ell \ind\{ \tau^\ell = \tau^\phi\} \vswitch  \right].
  \end{align}
  The first term inside the expectation denotes the total expected
  payoff until the agent chooses to switch, the second term denotes
  the payoff on switching.  Here, the expectation $\expec$ is
  conditioned on $(Z_0,N_0) = (z,n)$ and on the fact that agent $i$
  follows strategy $\phi$ and all other agents follow strategy
  $\xi$. (We drop this explicit dependence from the notation for
  $\expec$ for brevity.)  From this, we obtain,
  \begin{equation}
    \label{eq:vstayphi}
    \begin{split}
      \vstay^\phi(z,n)
      &= \sum_{\ell = 1}^\infty \gamma^{\ell-1} \expec\left[ F(Z_\ell,N_\ell) \ind\{ \tau^\ell \leq \tau^\phi\}\right] + \sum_{\ell = 1}^\infty \gamma^\ell  \vswitch \prob\left(\tau^\ell = \tau^\phi\right)\\
      &\leq \sum_{\ell = 1}^\infty \gamma^{\ell-1}
      \expec\left[  F(Z_\ell,N_\ell) \ind\{ \tau^\ell \leq
          \tau^\phi\}\right] + \gamma \vswitch.
  \end{split}
\end{equation}
Let $\nhat = \lfloor \sqrt{n-1}/2 + 1\rfloor$. For each $\ell=1,2, \cdots$, we have
\begin{equation}
  \label{eq:upper-bound-expec-reward-per-epoch}
  \begin{split}
    \expec\left[F(Z_\ell,N_\ell )\ind{\{\tau^\ell \leq \tau^\phi\}}
    \right]
    &= \expec\left[ \left. F(Z_\ell, N_\ell) \ind{\{\tau^\ell \leq \tau^\phi\}} \right| N_\ell \geq \nhat\right] \prob(N_\ell \geq \nhat )\\
    &\quad + \expec\left[ \left. F(Z_\ell, N_\ell)\ind{\{\tau^\ell \leq \tau^\phi \}} \right| N_\ell < \nhat\right] \prob(N_\ell < \nhat)  \\
    &\leq \frac{1}{2}(1-\gamma)^2 \lvswitch + \|F\|_\infty
    \prob(N_\ell < \nhat).
  \end{split}
\end{equation}
Here, in the inequality, the first term follows from the fact that
since $n > K_{\max}$, we have $\nhat > K_0$, and hence
$F(Z_\ell, N_\ell) < \tfrac{(1-\gamma)^2\lvswitch}{2}$ if
$N_\ell \geq \nhat$. In the second term, we have used the fact that
$F(Z_\ell, N_\ell) \leq \|F\|_\infty$.

To bound $\prob(N_\ell < \nhat)$, consider an auxiliary
system with $n$ agents at $t=0$ where each agent other than agent $i$
stays in the system for a time that is independently and identically
distributed as an exponential distribution with rate $\lambda$. (We
assume agent $i$ never leaves the auxiliary system.) Furthermore,
there are no arrivals to this auxiliary system. Let $\tilde{N}_t$
denote the number of agents in this auxiliary system. It is
straightforward to show that $\tilde{N}_t$ is first-order
stochastically dominated by $N_t$, via a coupling argument and we omit
the details here. This implies that
$\prob(N_\ell \leq \nhat) \leq \prob(\tilde{N}_\ell \leq \nhat)$,
where we write $\tilde{N}_\ell$ to denote $\tilde{N}_{\tau^\ell}$. Thus, we obtain
\begin{equation*}
  \expec\left[F(Z_\ell,N_\ell )\ind{\{\tau^\ell \leq \tau^\phi\}}\right] \leq \frac{1}{2}(1-\gamma)^2 \lvswitch + \|F\|_\infty \prob(\tilde{N}_\ell < \nhat).
\end{equation*}
Let $\lhat = \lfloor \sqrt{\log (n-1)} \rfloor$, and
$\that = \frac{\log (n-1)}{2 \lambda}$. For each $\ell \leq \lhat$, we
have $\tilde{N}_\ell \geq \tilde{N}_{\lhat}$, and hence,
\begin{equation}
  \begin{split}
    \label{eq:bound-on-n}
      \prob(\tilde{N}_{\ell} < \nhat)
      &\leq \prob(\tilde{N}_{\lhat}< \nhat) \\
      &= \prob(\tilde{N}_{\lhat} < \nhat |  \tau^{\lhat} < \that ) \prob(\tau^{\lhat} < \that) + \prob(\tilde{N}_{\lhat} < \nhat |  \tau^{\lhat} \geq \that ) \prob(\tau^{\lhat} \geq \that) \\
      &\leq \prob(\tilde{N}_{\lhat} < \nhat | \tau^{\lhat} < \that) \prob(\tau^{\lhat} < \that) + \prob(\tau^{\lhat} \geq \that ) \\
      &\leq \prob(\tilde{N}_{\that} < \nhat | \tau^{\lhat} < \that )\prob(\tau^{\lhat} < \that) + \prob(\tau^{\lhat} \geq \that) \\
      &\leq \prob(\tilde{N}_{\that} < \nhat) + \prob(\tau^{\lhat} \geq \that),
    \end{split}
  \end{equation}
  where the third inequality follows from the fact that on
  $\tau^{\lhat} < \that$, we have
  $\tilde{N}_{\that} \leq \tilde{N}_{\lhat}$, and the fourth
  inequality follows from the independence of $\tau^{\lhat}$ and
  $\tilde{N}_t$.

  Now, observe that since each agent $j \neq i$ stays in the auxiliary
  system for a time distributed independently and exponentially with
  rate $\lambda$, the probability that the agent $j \neq i$ is still
  in the auxiliary system by time $\that$ is equal to
  $\exp(-\lambda \that) = 1/\sqrt{n-1}$. Thus, the number of agents
  $\tilde{N}_{\that}$ in the auxiliary system at time $\that$ is
  distributed as $1 + \text{Bin}(n-1,\frac{1}{\sqrt{n-1}})$, where
  $\text{Bin}(\cdot, \cdot)$ denotes the binomial
  distribution. (Recall that in the auxiliary system, agent $i$ never
  leaves.)  Now, note that
  $\expec[\text{Bin}(n-1, \frac{1}{\sqrt{n-1}})] = \sqrt{n-1} >
  \nhat-1$. From this, we obtain
    \begin{equation}
          \label{eq:bound-on-ntilde}
      \begin{split}
      \prob(\tilde{N}_{\that} < \nhat)
      &= \prob\left(\text{Bin}\left(n-1,\frac{1}{\sqrt{n-1}}\right) < \nhat - 1\right)\\
      &\leq \prob\left(\text{Bin}\left(n-1,\frac{1}{\sqrt{n-1}}\right) < \frac{1}{2} \sqrt{n-1}\right)\\
      &\leq \exp\left(-\frac{1}{8} \sqrt{n-1}\right),
    \end{split}
  \end{equation}
  where we have used the Chernoff bound \citep{mitzenmacherU05} for the
  lower tail of the binomial distribution in the last inequality.

  Next, note that $\tau^\ell \sim \text{Gamma}(\ell,\lambda)$, since
  $\tau^\ell$ is the sum of $\ell$ independently and exponentially
  distributed time intervals. Hence, from Markov's inequality, we
  obtain
  \begin{equation}
    \label{eq:bound-on-tau}
    \prob(\tau^\lhat > \that) \leq \frac{\expec[\tau^\lhat ]}{\that} = \frac{\lhat}{\lambda \that}\leq \frac{2}{ \sqrt{\log(n-1)}}.
  \end{equation}
  Thus, combining \eqref{eq:upper-bound-expec-reward-per-epoch},
  \eqref{eq:bound-on-n}, \eqref{eq:bound-on-ntilde} and
  \eqref{eq:bound-on-tau}, we obtain for all $\ell \leq \lhat$,
  \begin{equation*}
    \begin{split}
      \expec\left[F(Z_\ell,N_\ell )\ind{\{\tau^\ell \leq
          \tau^\phi\}}\right] &\leq \frac{1}{2}(1-\gamma)^2 \lvswitch  +
      \|F\|_\infty \left(\exp(-\frac{1}{8}\sqrt{n-1}) +
        \frac{2}{ \sqrt{\log(n-1)}}\right).
    \end{split}
  \end{equation*}
  Thus, using \eqref{eq:vstayphi}, we have
  \begin{equation*}
    \begin{split}
    \vstay^\phi(z,n) =& \sum_{\ell = 1}^{\lhat} \gamma^{\ell -1}       \expec\left[F(Z_\ell,N_\ell )\ind{\{\tau^\ell \leq
    \tau^\phi\}}\right] + \sum_{\ell = \lhat + 1}^\infty     \gamma^{\ell -1}   \expec\left[F(Z_\ell,N_\ell )\ind{\{\tau^\ell \leq
    \tau^\phi\}}\right]  \\
    & + \gamma \vswitch\\
    \leq   & \frac{1}{1-\gamma} \left(\frac{1}{2}(1-\gamma)^2 \lvswitch +
    \|F\|_\infty \left(\exp(-\frac{1}{8}\sqrt{n-1}) +
      \frac{2}{ \sqrt{\log(n-1)}}\right)\right)\\
  & + \gamma^{\lhat} \|F\|_\infty  + \gamma \vswitch,
\end{split}
\end{equation*}
  where in the inequality, we use that fact that
  $F(Z_\ell, N_\ell) \leq \|F\|_\infty$ for all
  $\ell >\lhat$. Thus, we obtain,
  \begin{equation*}
    \begin{split}
    \vstay^\phi(z,n)     \leq &  \frac{1}{1-\gamma} \left(\frac{1}{2}(1-\gamma)^2 \lvswitch +
    \|F\|_\infty \left(\exp(-\frac{1}{8}\sqrt{n-1}) +
      \frac{2}{ \sqrt{\log(n-1)}}   + \gamma^{\lhat} (1 - \gamma) \right)\right) \\
      & + \gamma \vswitch.
\end{split}
\end{equation*}
Now, note that since $n \geq K_{\max} \geq K_1$, we have
\begin{align*} \|F\|_\infty \left(\exp(-\frac{1}{8}\sqrt{n-1}) +
    \frac{2}{ \sqrt{\log(n-1)}} + \gamma^{\lhat} (1 - \gamma) \right)
  < \frac{(1-\gamma)^2\lvswitch}{4}.
\end{align*}
Thus we obtain
$\vstay^\phi(z,n) \leq \frac{1}{1-\gamma} \left(
  \tfrac{(1-\gamma)^2\lvswitch}{2} +
  \frac{(1-\gamma)^2\lvswitch}{4}\right) + \gamma \vswitch =
\frac{3(1-\gamma)}{4} \lvswitch + \gamma \vswitch$.  Since this
inequality holds for all strategies $\phi \in \markovian$ and since
$\vswitch \geq \lvswitch$, we obtain
$\valst(z, n ; \xi, \vswitch) < \vswitch$ for all $z \in \setZ$ and
all $n \geq K_{\max}$.
\end{proof}

Let $\Upsilon = \compact \times [\lvswitch,\uvswitch]$. The preceding
lemma implies that for any $\zeta \in \mathcal{X}(\xi, \vswitch)$ with
$(\xi, \vswitch) \in \Upsilon$, it must be the case that
$\zeta(z,n) = 0$ for all $z \in \setZ$ and all $n \geq K_{\max}$. From
the definition of $\compact$, this implies that
$\mathcal{X}(\xi, \vswitch) \subseteq \compact$ for all
$(\xi, \vswitch) \in \Upsilon$. Thus, we can view the map
$(\xi, \vswitch) \to \mathcal{X}(\xi, \vswitch)$ as defining a
correspondence $\mathcal{X} : \Upsilon \rightrightarrows \compact$.

\section{Upper-hemicontinuity of $\map$}
\label{ap:upper-hemicontinuity}
For $(\xi, \vswitch) \in \Upsilon$, define the map $\map$ as
$\map(\xi, \vswitch) = \mathcal{X}(\xi, \vswitch) \times \{
\valsw(\xi, \vswitch)\}$. Note that from
Lemma~\ref{lem:upperbound-val}, Lemma~\ref{lem:lowerbound-val} and
Lemma~\ref{lem:convergence-vstay}, we obtain that
$\map(\xi, \vswitch) \subseteq \Upsilon$ for any
$(\xi, \vswitch) \in \Upsilon$. This implies that we can view the map
$\map$ as a correspondence
$\map : \Upsilon \rightrightarrows \Upsilon$. In this section, we seek
to show that this correspondence is upper-hemicontinuous. This result
directly used in proof for Theorem~\ref{thm:existence}.


To prove this, we first show that the value functions
$\valfn(\xi, \vswitch)$ and $\valsw(\xi, \vswitch)$ are jointly
continuous in $(\xi, \vswitch) \in \Upsilon$. In the following, we use
the following notation: for $U \in \bounded$, let
$\|U\|_* \defeq \max_{z \in \setZ, n < K_{\max}} |U(z,n)|$. Note that
$\|U\|_* \leq \|U\|_\infty$.
\begin{lemma}
  \label{lem:valf-continuous}
  The map $(\xi, \vswitch) \to \valfn(\xi, \vswitch)$ is (jointly)
  continuous in $(\xi, \vswitch) \in \Upsilon$.
\end{lemma}
\begin{proof}
  For $(\xi^i, \vswitch^i) \in \Upsilon$ for $i=1,2$, let
  $W_i(z,n) = \valfn(z,n; \xi^i, \vswitch^i)$. Using the definition of
  $\dynamic$ and Lemma~\ref{lem:convergence-vstay}, we obtain
  $W_i(z,n) = F(z,n) + \gamma \vswitch^i$ for all $z \in \setZ$ and
  $n \geq K_{\max}$. This implies that
  \begin{align*}
    | W_1(z,n) - W_2(z,n) | \leq \gamma |\vswitch^1 - \vswitch^2|, \quad \text{for $z \in \setZ$ and $n \geq K_{\max}$.}
  \end{align*}
  This implies that
  $\|W_1 - W_2\|_\infty \leq \max\{\|W_1 - W_2\|_* , \gamma |
  \vswitch^1 - \vswitch^2|\}$.
  Next, we have
  \begin{align*}
    \|W_1  - W_2\|_*
   = &  \| \dynamic(\xi^1, \vswitch^1,W_1) -\dynamic(\xi^2, \vswitch^2,W_2)\|_*\\
    \leq &  \| \dynamic(\xi^1, \vswitch^1,W_1) -\dynamic(\xi^2, \vswitch^2,W_1)\|_* \\
    & +  \| \dynamic(\xi^2, \vswitch^2,W_1) -\dynamic(\xi^2, \vswitch^2,W_2)\|_*\\
    \leq  & \| \dynamic(\xi^1, \vswitch^1,W_1) -\dynamic(\xi^2, \vswitch^2,W_1)\|_*  \\
    & +  \| \dynamic(\xi^2, \vswitch^2,W_1) -\dynamic(\xi^2, \vswitch^2,W_2)\|_\infty\\
    \leq  & \| \dynamic(\xi^1, \vswitch^1,W_1) -\dynamic(\xi^2, \vswitch^2,W_1)\|_* + \gamma \| W_1 -W_2\|_\infty.
  \end{align*}
  where we have used Lemma~\ref{lem:bellman-contraction} in the last
  inequality. Using the fact that
  $\|W_1 - W_2\|_* \leq \|W_1 - W_2\|_\infty \leq \max\{\|W_1 -
  W_2\|_* , \gamma | \vswitch^1 - \vswitch^2|\} \leq \|W_1 - W_2\|_* +
  \gamma | \vswitch^1 - \vswitch^2|$ and after some straightforward
  algebra, we obtain
  \begin{align*}
    \|W_1 - W_2\|_\infty
    &\leq \frac{1}{1-\gamma} \left( \| \dynamic(\xi^1, \vswitch^1,W_1) -\dynamic(\xi^2, \vswitch^2,W_1)\|_* + \gamma |\vswitch^1 - \vswitch^2|\right).
  \end{align*}
  From Lemma~\ref{lem:expec-continuous}, we obtain that the first term
  in the parenthesis can be made arbitrarily small by setting
  $\| \xi^1 - \xi^2 \|_\infty$ and $|\vswitch^1 -\vswitch^2|$
  correspondingly small enough. Thus, we conclude that
  $\valfn(\xi, \vswitch)$ is jointly continuous in
  $(\xi, \vswitch) \in \Upsilon$.
\end{proof}
The following auxiliary lemma is used in the proof of
Lemma~\ref{lem:valf-continuous}.
\begin{lemma}\label{lem:expec-continuous} Let $(\xi^m, \vswitch^m) \in \Upsilon$ with $(\xi^m, \vswitch^m) \to (\xi, \vswitch) \in \Upsilon$ as $m \to \infty$. For any $U \in \bounded$, we
  have
  $\|\dynamic(\xi^m, \vswitch^m, U) - \dynamic(\xi, \vswitch, U)\|_*
  \to 0$ as $m \to \infty$.
\end{lemma}
\begin{proof} Let $(\xi^m, \vswitch^m)$ be as in the statement of the
  lemma, and let $W_m = \dynamic(\xi^m, \vswitch^m, U)$ and
  $W = \dynamic(\xi, \vswitch, U)$. By definition of $\dynamic$, we
  have
  \begin{align*}
    |W_m(z,n) - W(z,n)|
    = &  \gamma | \max\{ \expec^m[ U(Z_\tau, N_\tau) | z,n], \vswitch^m\}  \\
    & - \max\{ \expec^\xi[ U(Z_\tau, N_\tau) | z,n], \vswitch\} |\\
    \leq & \gamma \max\{ |\expec^m[ U(Z_\tau, N_\tau) | z,n] \\
    & - \expec^\xi[ U(Z_\tau, N_\tau) | z,n] |, | \vswitch^m  -  \vswitch|\},
  \end{align*}
  where we let $\expec^m = \expec^{\xi^m}$. Thus, it suffices to show
  that the first term inside the maximization converges to zero as
  $m \to \infty$ for all $z \in \setZ$ and $n < K_{\max}$.  Observe
  that, since $U \in\bounded$ and $\tau$ is exponentially distributed
  with parameter $\lambda$, we have
  \begin{align*}
    \expec^\xi\left[ U(Z_\tau, N_\tau) | z,n \right]
   &= \int_0^\infty \lambda \exp(- \lambda t) \expec^\xi[ U(Z_t, N_t) | z,n, \tau = t] dt\\
   &= \int_0^T \lambda \exp(- \lambda t) \expec^\xi[ U(Z_t, N_t) | z,n, \tau = t] dt\\
    &\quad + \int_T^\infty \lambda \exp(- \lambda t) \expec^\xi[ U(Z_t, N_t) | z,n, \tau = t] dt
  \end{align*}
  with similar expressions for $\xi^m$ in place of $\xi$. For large
  enough value of $T>0$, the second term in the last equation can be
  made arbitrarily small (uniformly for $\xi$ and all $\xi^m$). Thus,
  again it suffices to show that the first term in the last equation
  is continuous in $(\xi, \vswitch)$ for all $z \in \setZ$ and
  $n < K_{\max}$ and for large enough $T$.

  Now, using the definition \eqref{eq:loc-dyn} of the transition rate
  matrix $\sQ^\xi$ of the chain $\markov(\xi, \kappa(\xi))$ (and
  similarly $\sQ^m = \sQ^{\xi^m}$ of the chain
  $\markov(\xi^m, \kappa(\xi^m))$), we obtain that
  $\sQ^m( (u,k) \to (v,\ell)) \longrightarrow \sQ( (u,k) \to
  (v,\ell))$ as $m \to \infty$ for all $(u,k), (v,\ell) \in
  \sets$. Then, from \citep[See pg. 2183, Example 1.1]{xia94} or
  \citep[pg. 262, problem 8]{ethierK86}, we obtain that the measure
  $\prob^m( \cdot | z,n, \tau=t)$ converges weakly to
  $\prob^{\xi}(\cdot | z,n, \tau=t)$. From this, we conclude that
  $\int_0^T \lambda \exp(- \lambda t) \expec^m[ U(Z_t, N_t) | z,n,
  \tau = t] dt$ converges to
  $\int_0^T \lambda \exp(- \lambda t) \expec^\xi[ U(Z_t, N_t) | z,n,
  \tau = t] dt$ as $m \to \infty$. This completes the proof.
\end{proof}
The continuity of $\valsw(\xi, \vswitch)$ is then obtained as a
corollary of Lemma~\ref{lem:valf-continuous}.
\begin{lemma}\label{lem:valsw-continuous} The value function 
  $\valsw(\xi, \vswitch)$ is jointly continuous in
  $(\xi, \vswitch) \in \Upsilon$.
\end{lemma}
\begin{proof}

  Recall the definition \eqref{eq:stswdef} of $\valsw(\xi,
  \vswitch)$:
  \begin{align*}
    \valsw(\xi, \vswitch) = \sum_{(z,n) \in \sets} \pi(z,n) \valsw(z,n+1; \xi, \vswitch),
  \end{align*}
  where $\pi = \pi(\xi, \kappa(\xi))$ is the invariant distribution of
  $\markov(\xi, \kappa(\xi))$. From Lemma~\ref{lem:upperbound-val}, we
  have $\|\valst(\xi, \vswitch)\|_\infty \leq \uvswitch$. Also, note
  that Lemma~\ref{lem:stationary-dist-continuity} and
  Lemma~\ref{lem:continuity-kappa} imply that the invariant
  distribution $\pi(\xi, \kappa(\xi))$ is continuous. Moreover, from
  Lemma~\ref{lem:tightness}, we obtain that the set of invariant
  distributions $\Gamma$ is tight. These results together imply that
  it suffices to show that $\valst(z,n; \xi, \vswitch)$ is uniformly
  continuous in $(\xi, \vswitch) \in \Upsilon$ for all $z \in \setZ$
  and all $n < M$ for some large enough $M$.

  Let $(\xi^m, \vswitch^m) \in \Upsilon$ with
  $(\xi^m, \vswitch^m) \to (\xi, \vswitch) \in \Upsilon$ as $m \to \infty$. We have
  \begin{align*}
    & | \valst(z,n; \xi^m, \vswitch^m) - \valst(z,n; \xi, \vswitch) |\\
    \leq & | \expec^{\xi^m}[ \valfn(Z_\tau, N_\tau;\xi^m, \vswitch^m)| z,n] - \expec^{\xi}[ \valfn(Z_\tau, N_\tau;\xi, \vswitch)| z,n] |\\
    \leq & | \expec^{\xi^m}[ \valfn(Z_\tau, N_\tau;\xi^m, \vswitch^m)| z,n] - \expec^{\xi^m}[ \valfn(Z_\tau, N_\tau;\xi, \vswitch)| z,n] |\\
    & +  | \expec^{\xi^m}[ \valfn(Z_\tau, N_\tau;\xi, \vswitch)| z,n] - \expec^{\xi}[ \valfn(Z_\tau, N_\tau;\xi, \vswitch)| z,n] |\\
    \leq  &|  \valfn(\xi^m, \vswitch^m) - \valfn(\xi, \vswitch)\|_\infty +  | \expec^{\xi^m}[ \valfn(Z_\tau, N_\tau;\xi, \vswitch)| z,n]\\
    & - \expec^{\xi}[ \valfn(Z_\tau, N_\tau;\xi, \vswitch)| z,n] |.
  \end{align*}
  From Lemma~\ref{lem:valf-continuous}, we obtain that as
  $m \to \infty$, the first term converges to zero. Moreover, from the
  same argument as in the proof of Lemma~\ref{lem:valf-continuous}, we
  obtain that
  $\expec^{\xi^m}[ \valfn(Z_\tau, N_\tau;\xi, \vswitch)| z,n] \to
  \expec^{\xi}[ \valfn(Z_\tau, N_\tau;\xi, \vswitch)| z,n]$ as
  $m \to \infty$ for each $(z,n) \in \sets$. From this, we conclude
  that $\valst(z,n;\xi, \vswitch)$ is uniformly continuous in
  $(\xi, \vswitch) \in \Upsilon$ for all $z \in \setZ$ and all $n < M$
  for large enough $M$.
\end{proof}


We are now ready to show that the correspondence $\map$ is
upper-hemicontinuous.
\begin{lemma}\label{lem:upperhemicontinuous} The correspondence
  $\map : \Upsilon \rightrightarrows \Upsilon$ is
  upper-hemicontinuous.
\end{lemma}
\begin{proof}

  By definition,
  $\map(\xi, \vswitch) = \mathcal{X}(\xi, \vswitch) \times \{
  \valsw(\xi, \vswitch)\}$ for $(\xi, \vswitch) \in \Upsilon$. From
  Lemma~\ref{lem:valsw-continuous}, we obtain that
  $\valsw(\xi, \vswitch)$ is jointly continuous in
  $(\xi, \vswitch) \in \Upsilon$. Thus, it suffices to show that the
  correspondence $\mathcal{X}: \Upsilon \rightrightarrows \compact$ is
  upper-hemicontinuous.

  Consider a sequence
  $(\xi^n, \vswitch^n , \zeta^n) \to (\xi, \vswitch, \zeta)$ as
  $n \to \infty$ such that
  $\zeta^n \in \mathcal{X}(\xi^n, \vswitch^n)$ for each $n \geq 0$. By
  continuity of $\valst(\cdot)$, we obtain that if
  $\valst(z,n, \xi, \vswitch) > \vswitch$ for some $(z,n) \in \sets$,
  then for all large enough $m$, we must have
  $\valst(z,n, \xi^m, \vswitch^m) > \vswitch^m$, and hence
  $\zeta^m(z,n) = 1$. Similarly, if
  $\valst(z,n, \xi, \vswitch) < \vswitch$, then $\zeta^m(z,n) = 0$ for
  all large enough $m$. Since $\zeta^m \to \zeta$, this implies that
  $\zeta(z,n) = 1$ if $\valst(z,n, \xi, \vswitch) > \vswitch$, and
  $\zeta(z,n) = 0$ if $\valst(z,n, \xi, \vswitch) < \vswitch$. Thus,
  we obtain that $\zeta(z,m) \in \mathcal{X}(\xi, \vswitch)$.
\end{proof}

\section{Existence of an optimal threshold strategy}
\label{ap:proofs_threshold_strategy}

In this section we provide the proof of
Theorem~\ref{thm:threshold}. We prove this result in two steps: first,
we prove Lemma~\ref{lem:value-decreasing}, which states that the value
function $\vstay: \sets \rightarrow \reals$ is non-increasing in the
number of agents $n$ at the location for any fixed resource level
$z\ in \setZ$. Second, we show in Lemma~\ref{lem:convergence-vstay}
that $\lim_{n \rightarrow \infty}\vstay(z,n) \leq \gamma\vswitch$ for
all $z \in \setZ$. Therefore there always exists a threshold strategy
in the set of best responses $\opt(\xi, \kappa, \vswitch)$.

\begin{proof}[Proof of Lemma \ref{lem:value-decreasing}]
  We define a partial order $\preccurlyeq_p$ on the state space
  $\sets$ of $\markov(\xi, \kappa)$ as follows: for
  $(z_1, n_1), (z_2, n_2) \in \sets$,
  $(z_1, n_1) \preccurlyeq_p (z_2, n_2)$ if and only if $z_1 = z_2$
  and $n_1 \leq n_2$. For any function $f: \sets \rightarrow \reals$,
  we say $f$ is decreasing with respect to $\preccurlyeq_p$ if for all
  $(z_1, n_1), (z_2, n_2) \in \sets$ such that
  $(z_1, n_1) \preccurlyeq_p (z_2, n_2)$, we have
  $f(z_1, n_1) \geq f(z_2, n_2)$.

  Thus, our goal is to show that the value function $\vstay$ of
  $\optstop(\xi,\kappa, \vswitch)$ is decreasing with respect to
  $\preccurlyeq_p$. We note that the property ``decreasing with
  respect to $\preccurlyeq_p$'' is a closed convex cone property for
  functions on $\sets$, as defined in \citep{smithM2002}. Thus, using
  their Proposition~5, we can conclude that $\vstay$ has this property
  if the following two conditions hold:
  \begin{enumerate}
  \item The resource sharing function $F$ is decreasing with respect
    to $\preccurlyeq_p$.
  \item Let $(Z_t, N_t) \sim \markov(\xi, \kappa)$. For any
    $(z,n) \in \sets$, let $\nu_{(z,n)}$ be the probability
    distribution of $(Z_{\tau}, N_{\tau}) \sim \nu_{(z,n)}$
    conditioning on $(Z_0, N_0) = (z,n)$, where $\tau$ is distributed
    independently as an exponential with rate $\lambda$, denoting the
    first decision epoch of a fixed agent. Then for any
    $f: \sets \rightarrow \reals$ that is decreasing with respect to
    $\preccurlyeq_p$, it must hold that
    $\expec[f(Z, N) | (Z,N) \sim \nu_{(z_1,n_1)}] \geq \expec[f(Z, N)
    | (Z,N) \sim \nu_{(z_2, n_2)}]$ for all
    $(z_1, n_1) \preccurlyeq_p (z_2, n_2)$.
  \end{enumerate}
  Since $F(z,n)$ is decreasing in $n$ for each $z \in \setZ$, we
  immediately obtain the first condition. We now show that the second
  condition also holds using a coupling argument.

  Suppose $(z_1, n_1) \preccurlyeq_p (z_2, n_2)$. By using an argument
  same as that in the proof of Lemma~\ref{lem:coupling}, we obtain
  that there exists a coupling of the two processes
  $(\zstate{t}{i}, \nagent{t}{i}) \sim \markov(\xi, \kappa)$ with
  $(\zstate{0}{i}, \nagent{0}{i}) = (z_i, n_i)$ for $i=1,2$, such that
  for all $t \geq 0$,
  $(\zstate{t}{1}, \nagent{t}{1}) \preccurlyeq_p (\zstate{t}{2},
  \nagent{t}{2})$. Thus, for any $f$ that is decreasing with respect
  to $\preccurlyeq_p$ we have
  $f(\zstate{t}{1}, \nagent{t}{1}) \geq f(\zstate{t}{2},
  \nagent{t}{2})$ for all $t \geq 0$, and therefore
  $ \expec[f(\zstate{\tau}{1}, \nagent{\tau}{1})] \geq
  \expec[f(\zstate{\tau}{2}, \nagent{\tau}{2})]$, where $\tau$ is a
  distributed independently as an exponential with rate
  $\lambda$. Since
  $(\zstate{\tau}{i}, \nagent{\tau}{i}) \sim \nu_{(z_i, n_i)}$ for
  $i=1,2$, we obtain the result.
\end{proof}

\section{Coupling results}
\label{ap:coupling}

In this section, we obtain structural properties of the Markov chain
$\markov(\xi, \kappa)$ by coupling the chain with an $M/M/\infty$
queue.

Let $\mminfty{\lambda}{\kappa}$ denote an (independent) $M/M/\infty$
queue with arrival rate $\kappa$ and service rate $\lambda$. We begin
with the following simple result which states that a queue with higher
arrival rate and/or lower service rate is more likely to have more
agents in the queue. The proof is straightforward and omitted.

\begin{lemma}
  \label{lem:compare-mminfty-queue-arrival}
  Let $\nagent{t}{i}$, for $i=1,2$, denote the number of agents at
  time $t$ in an (independent) $M/M/\infty$ queue with arrival rate
  $\kappa_i$ and service rate $\lambda_i$. Suppose
  $\nagent{0}{1} = \nagent{0}{2}$, and one of the following two
  conditions holds: (1) $\lambda_1 = \lambda_2$ and
  $\kappa_1 \leq \kappa_2$; or (2) $\lambda_1 \geq \lambda_2$ and
  $\kappa_1 = \kappa_2$. Then, for all $t \geq 0$, $\nagent{t}{1}$ is
  stochastically dominated by $\nagent{t}{2}$, i.e., for all
  $n \in \naturals_0$, we have
  $\prob(\nagent{t}{1} \geq n) \leq \prob(\nagent{t}{2} \geq n)$.
\end{lemma}




In the proof of Theorem~\ref{thm:existence}, we frequently compare the
$\markov(\xi, \kappa)$ process for two (or more) different values of
$(\xi, \kappa)$ to show the monotonicity of various
quantities. Our next result justifies these stochastic
comparisons. Before we state the lemma, we make the following
definition of stochastic dominance. Let
$\lstate{t}{i} \sim \markov(\xi^i, \kappa^i)$ with
$\lstate{0}{i} = (z_i,n_i)$ for $i=1,2$. We say the process
$\lstate{t}{1}$ is stochastically dominated by the process
$\lstate{t}{2}$ if
\begin{align*}
  \prob(\zstate{t}{1} = z , \nagent{t}{1} \geq n) \leq   \prob(\zstate{t}{2} = z , \nagent{t}{2} \geq n), \quad \text{for all $(z,n) \in \sets$.}
\end{align*}
In that case, we denote as $\lstate{t}{1} \sdleq \lstate{t}{2}$. Note
that this also implies that $\nagent{t}{1}$ is stochastically
dominated by $\nagent{t}{2}$ under the usual sense of stochastic
dominance.



\begin{lemma}
  \label{lem:coupling}
  Let $\xi \in \markovian$ and $\kappa >0$. Let
  $(Z_t,N_t) \sim \markov(\xi, \kappa)$.
  \begin{enumerate}
  \item Let $\kappa_0 \geq \kappa$, and let $\xi_0 \in \markovian$ be
    such that $\xi_0(z,n) \geq \xi(z,n)$ for all $(z,n) \in
    \sets$. Then we have $(Z_t,N_t) \sdleq \lstate{t}{0}$ for all
    $t \geq 0$, where $\lstate{t}{0} \sim \markov(\xi_0, \kappa_0)$
    with $\zstate{0}{0} = Z_0$ and $\nagent{0}{0} \geq N_0$.
  \item Let $X^i_t \sim \mminfty{\lambda_i}{\kappa}$ for
    $i = 1, 2$ be two independent processes with
    $X_0^1 = X_0^2 = N_0$, where $\lambda_1 = \lambda$
    and $\lambda_2 = (1-\gamma)\lambda$. Then, we have for all
    $t\geq 0$,
    $(Z_t,X_t^1) \sdleq (Z_t,N_t) \sdleq (Z_t,X_t^2)$.
  \end{enumerate}
\end{lemma}
\begin{proof} First note that the second statement in the lemma is
  implied by the first. In particular, let $\xi_1(z,n) = 0$ and
  $\xi_2(z,n) = 1$ for all $(z,n) \in \sets$. Then, using the first
  statement in the lemma, we obtain
  $\lstate{t}{1} \sdleq (Z_t,N_t) \sdleq \lstate{t}{2}$, where
  $\lstate{t}{i} \sim \markov(\xi_i,\kappa)$ with
  $\lstate{0}{i} = (Z_0,N_0)$. The second statement then follows
  directly by the fact that under $\xi_i$, the process $\lstate{t}{i}$
  has the same distribution as $(Z_t, X_t^i)$ for $i=1,2$.

  To prove the first statement in the lemma, we use a coupling
  argument. We construct two chains as follows. Let $t_0= 0$ and
  $(Z_0,N_0) = (z_0,n_0)$ and $\lstate{0}{0} = (u_0,v_0)$, with
  $u_0= z_0$ and $v_0 \geq n_0$. For $k=1,2 \ldots,$ define the
  following recursively:

  \begin{enumerate}
  \item Let $\tau_k \sim \mathsf{Exp}(\Delta_k)$ where
    $\Delta_k \defeq \sum_{y \neq u_{k-1}} \mu_{u_{k-1},y} + \kappa_0 +
    \lambda v_{k-1}$. Let $t_k = t_{k-1} + \tau_k$.

  \item Let $\lstate{t}{0} = (u_{k-1}, v_{k-1})$ and
    $(Z_t, N_t) = (z_{k-1},n_{k-1})$ for $t \in [t_{k-1}, t_k)$.

  \item For $t = t_k$, let $\lstate{t}{0} = (u_k,v_k)$, where
    \begin{align*}
      (u_k, v_k) = \begin{cases} (y,v_{k-1}) \\
      \qquad \text{with probability $\mu_{u_{k-1},y}/\Delta_k$, for each $y \in \setZ$ with $y \neq u_{k-1}$;}\\
        (u_{k-1},v_{k-1} + 1) \\
        \qquad \text{with probability $\kappa_0/\Delta_k$;}\\
        (u_{k-1},v_{k-1}  - 1) \\
        \qquad  \text{with probability $\lambda v_{k-1} (1 - \gamma \xi_0(u_{k-1}, v_{k-1}))/\Delta_k$;}\\
        (u_{k-1}, v_{k-1}) \\
        \qquad \text{with probability $\lambda v_{k-1} \gamma \xi_0(u_{k-1}, v_{k-1})/\Delta_k$.}
      \end{cases} 
    \end{align*}

  \item Define
    $\zeta_k \defeq \frac{ n_{k-1} (1 - \gamma \xi(z_{k-1},
      n_{k-1}))}{ v_{k-1} (1 - \gamma \xi_0(u_{k-1}, v_{k-1}))}$ and
    $\eta_k \defeq \left(\frac{1 - \gamma \xi_0(u_{k-1},v_{k-1})}{\gamma
        \xi_0(u_{k-1},v_{k-1})}\right) \max(\zeta_k - 1,0)$.  It is
    straightforward to verify that $\eta_k \in [0,1]$.

  \item Let $(Z_t,N_t) = (z_k,n_k)$ for $t=t_k$, where
    \begin{align*}
      (z_k,n_k) = \begin{cases}
        (u_k, n_{k-1}) & \text{if $u_k \neq u_{k-1}$;}\\
        \begin{cases} (z_{k-1}, n_{k-1}+1) \\
        \quad \text{with probability $\frac{\kappa}{\kappa_0}$;}\\
          (z_{k-1}, n_{k-1}) \\
          \quad \text{with probability
            $1 -\frac{\kappa}{\kappa_0}$,}
        \end{cases}
        & \text{if $(u_k,v_k) = (u_{k-1}, v_{k-1} + 1)$;}\\
        \begin{cases}
          (z_{k-1},n_{k-1}-1)  \\ 
          \quad \text{with probability} \\
          \quad \text{$\min(\zeta_k,1)$;}\\
          (z_{k-1},n_{k-1})  \\
          \quad \text{with probability}\\
          \quad \text{$\max(1 - \zeta_k,0)$,}
        \end{cases}
        & \text{if $(u_k,v_k)  = (u_{k-1}, v_{k-1}-1)$;}\\
        \begin{cases}
          (z_{k-1}, n_{k-1}-1) \\
          \quad \text{with probability $\eta_k$;}\\
          (z_{k-1}, n_{k-1})  \\
          \quad \text{with probability $1- \eta_k$,}
        \end{cases}
        & \text{if $(u_k,v_k)  = (u_{k-1}, v_{k-1})$.}
      \end{cases}
    \end{align*}

  \end{enumerate}

  It is straightforward to verify that under this construction, we
  have $\lstate{t}{0} \sim \markov(\xi_0, \kappa_0)$ and
  $(Z_t,N_t) \sim \markov(\xi, \kappa)$ with $\zstate{0}{0} = Z_0 $
  and $N_0 \leq \nagent{0}{0} $. Furthermore, by construction, we have
  $z_k = u_k$ for all $k$, and hence $\zstate{t}{0} = Z_t$ for all
  $t \geq 0$.

  To show that $N_t \leq \nagent{t}{0}$ for all $t \geq 0$, we perform
  induction on $k$ in the above construction. Note that
  $n_0 \leq v_0$. Suppose for some $k$, we have
  $n_{k-1} \leq v_{k-1}$. Then, from the definition, we obtain that
  $n_k \leq v_k$ for all the cases, except possibly when
  $(u_k,v_k) = (u_{k-1},v_{k-1}-1)$ and
  $(z_k,n_k) = (z_{k-1},n_{k-1})$. Under this case, if
  $n_{k-1} < v_{k-1}$, then again we have $n_k \leq v_k$. On the other
  hand, if $n_{k-1} = v_{k-1}$, then together with the fact that
  $z_{k-1} = u_{k-1}$, we obtain
  $\xi(z_{k-1},n_{k-1}) \leq \xi_0(u_{k-1}, v_{k-1})$, implying that
  $\zeta_k \geq 1$. However, note that if $\zeta_k \geq 1$, then the
  event where $(u_k,v_k)=(u_{k-1},v_{k-1}-1)$ and
  $(z_k,n_k) = (z_{k-1},n_{k-1})$ occurs with zero probability. Thus,
  we obtain that under all cases, $n_k \leq v_k$. This completes the
  induction step and hence the proof.
\end{proof}

\bibliographystyle{authordate1}
\bibliography{mfe_exploration_arxiv}

\begin{thebibliography}{}

\bibitem[\protect\citename{NYC, }n.d.]{NYCtrip}
{\em New York City Taxi \& Limousine Comission Trip Record Data.}
\newblock \url{http://www.nyc.gov/html/tlc/html/about/trip_record_data.shtml}.

\bibitem[\protect\citename{Wea, }n.d.]{Weather}
{\em Weather History and Data Archive.}
\newblock \url{https://www.wunderground.com/history/}.

\bibitem[\protect\citename{Adlakha {\em et~al.\ }\relax, }2015]{sachin_2010}
Adlakha, Sachin, Johari, Ramesh, \& Weintraub, Gabriel~Y. 2015.
\newblock Equilibria of dynamic games with many players: Existence,
  approximation, and market structure.
\newblock {\em Journal of Economic Theory}, {\bf 156}(March), 269--316.

\bibitem[\protect\citename{Aliprantis \& Border, }2006]{aliprantisB2006}
Aliprantis, Charalambos~D., \& Border, Kim. 2006.
\newblock {\em Infinite Dimensional Analysis: A Hitchhiker's Guide}.
\newblock Springer Science \& Business Media.

\bibitem[\protect\citename{Arnosti {\em et~al.\ }\relax, }2014]{arnosti}
Arnosti, Nick, Johari, Ramesh, \& Kanoria, Yash. 2014.
\newblock Managing Congestion in Decentralized Matching Markets.
\newblock {\em Pages  451--451 of:} {\em Proceedings of the Fifteenth ACM
  Conference on Economics and Computation}.
\newblock EC '14.
\newblock New York, NY, USA: ACM.

\bibitem[\protect\citename{Arthur, }1994]{arthur1994inductive}
Arthur, W.~Brian. 1994.
\newblock Inductive reasoning and bounded rationality.
\newblock {\em The American economic review},  406--411.

\bibitem[\protect\citename{Balseiro {\em et~al.\ }\relax, }2015]{balseiro2014}
Balseiro, Santiago~R., Besbes, Omar, \& Weintraub, Gabriel~Y. 2015.
\newblock Repeated Auctions with Budgets in Ad Exchanges: Approximations and
  Design.
\newblock {\em Management Science}, {\bf 61}(4), 864--884.

\bibitem[\protect\citename{Banerjee {\em et~al.\ }\relax, }2015]{banerjeeJR15}
Banerjee, Siddhartha, Johari, Ramesh, \& Riquelme, Carlos. 2015.
\newblock Pricing in Ride-Sharing Platforms: A Queueing-Theoretic Approach.
\newblock {\em Pages  639--639 of:} {\em Proceedings of the Sixteenth ACM
  Conference on Economics and Computation}.
\newblock EC '15.
\newblock New York, NY, USA: ACM.

\bibitem[\protect\citename{Banerjee {\em et~al.\ }\relax, }2016]{banerjeeFL16}
Banerjee, Siddhartha, Freund, Daniel, \& Lykouris, Thodoris. 2016.
\newblock Pricing and Optimization in Shared Vehicle Systems: An Approximation
  Framework.
\newblock {\em CoRR}, {\bf abs/1608.06819}.

\bibitem[\protect\citename{Berge, }1963]{berge1963topological}
Berge, Claude. 1963.
\newblock {\em Topological Spaces: Including a treatment of multi-valued
  functions, vector spaces, and convexity, translated by {E.} {M.} Patterson}.
\newblock Dover.

\bibitem[\protect\citename{Billingsley, }2013]{billingsley2013}
Billingsley, Patrick. 2013.
\newblock {\em Convergence of probability measures}.
\newblock John Wiley \& Sons.

\bibitem[\protect\citename{Bimpikis {\em et~al.\ }\relax, }2016]{bimpikisCD16}
Bimpikis, Kostas, Candogan, Ozan, \& Daniela, Saban. 2016.
\newblock Spatial Pricing in Ride-Sharing Networks.
\newblock {\em Available at SSRN: https://ssrn.com/abstract=2868080}.

\bibitem[\protect\citename{Braverman {\em et~al.\ }\relax,
  }2016]{bravermanDLY16}
Braverman, Anton, Dai, J.G., Liu, Xin, \& Ying, Lei. 2016.
\newblock Empty-car routing in ridesharing systems.
\newblock {\em arXiv preprint arXiv:1609.07219}.

\bibitem[\protect\citename{Castillo {\em et~al.\ }\relax, }2017]{CKW17}
Castillo, Juan~Camilo, Knoepfle, Dan, \& Weyl, Glen. 2017.
\newblock Surge Pricing Solves the Wild Goose Chase.
\newblock {\em Pages  241--242 of:} {\em Proceedings of the 2017 ACM Conference
  on Economics and Computation}.
\newblock EC '17.
\newblock New York, NY, USA: ACM.

\bibitem[\protect\citename{Chakrabarti {\em et~al.\ }\relax,
  }2009]{chakrabarti2009kolkata}
Chakrabarti, Anindya-Sundar, Chakrabarti, Bikas~K., Chatterjee, Arnab, \&
  Mitra, Manipushpak. 2009.
\newblock The Kolkata Paise Restaurant problem and resource utilization.
\newblock {\em Physica A: Statistical Mechanics and its Applications}, {\bf
  388}(12), 2420--2426.

\bibitem[\protect\citename{Chakrabarti, }2007]{chakrabarti2007kolkata}
Chakrabarti, Bikas~K. 2007.
\newblock Kolkata restaurant problem as a generalised El Farol Bar problem.
\newblock {\em Pages  239--246 of:} {\em Econophysics of Markets and Business
  Networks}.
\newblock Springer.

\bibitem[\protect\citename{Chen, }2016]{Chen2016}
Chen, M.~Keith. 2016.
\newblock Dynamic Pricing in a Labor Market: Surge Pricing and Flexible Work on
  the Uber Platform.
\newblock {\em Pages  455--455 of:} {\em Proceedings of the 2016 ACM Conference
  on Economics and Computation}.
\newblock EC '16.
\newblock New York, NY, USA: ACM.

\bibitem[\protect\citename{Durrett \& Levin, }1994a]{durrettL94}
Durrett, R., \& Levin, S. 1994a.
\newblock The Importance of Being Discrete (and Spatial).
\newblock {\em Theoretical Population Biology}, {\bf 46}(3), 363 -- 394.

\bibitem[\protect\citename{Durrett \& Levin, }1994b]{durrettL94b}
Durrett, Richard, \& Levin, Simon~A. 1994b.
\newblock Stochastic Spatial Models: A User{\textquoteright}s Guide to
  Ecological Applications.
\newblock {\em Philosophical Transactions of the Royal Society of London B:
  Biological Sciences}, {\bf 343}(1305), 329--350.

\bibitem[\protect\citename{Ethier \& Kurtz, }1986]{ethierK86}
Ethier, Stewart~N., \& Kurtz, Thomas~G. 1986.
\newblock {\em Markov processes: characterization and convergence}.
\newblock John Wiley \& Sons.

\bibitem[\protect\citename{Fudenberg \& Tirole, }1991a]{tirole}
Fudenberg, Drew, \& Tirole, Jean. 1991a.
\newblock {\em Game theory}.
\newblock Cambridge, Massachusetts: The MIT Press.

\bibitem[\protect\citename{Fudenberg \& Tirole, }1991b]{fudenbergT91}
Fudenberg, Drew, \& Tirole, Jean. 1991b.
\newblock Perfect Bayesian Equilibrium and Sequential Equilibrium.
\newblock {\em Journal of Economic Theory}, {\bf 53}(2), 236--260.

\bibitem[\protect\citename{Ghosh {\em et~al.\ }\relax,
  }2010]{ghosh2010statistics}
Ghosh, Asim, Chatterjee, Arnab, Mitra, Manipushpak, \& Chakrabarti, Bikas~K.
  2010.
\newblock Statistics of the kolkata paise restaurant problem.
\newblock {\em New Journal of Physics}, {\bf 12}(7), 075033.

\bibitem[\protect\citename{Hopenhayn, }1992]{hopenhayn_1992}
Hopenhayn, H.~A. 1992.
\newblock Entry, Exit and Firm Dynamics in Long Run Equilibrium.
\newblock {\em Econometrica}, {\bf 60}(5), 1127 -- 1150.

\bibitem[\protect\citename{Huang {\em et~al.\ }\relax, }2007]{huang_2007}
Huang, M., Caines, P.~E., \& Malham\'{e}, R.~P. 2007.
\newblock Large-Population Cost-Coupled {LQG} Problems With Nonuniform Agents:
  Individual-Mass Behavior and Decentralized $\epsilon$-{Nash} Equilibria.
\newblock {\em IEEE Transactions on Automatic Control}, {\bf 52}(9),
  1560--1571.

\bibitem[\protect\citename{Iyer {\em et~al.\ }\relax, }2014]{iyerJS2014}
Iyer, Krishnamurthy, Johari, Ramesh, \& Sundararajan, Mukund. 2014.
\newblock Mean Field Equilibria of Dynamic Auctions with Learning.
\newblock {\em Management Science}, {\bf 60}(12), 2949--2970.

\bibitem[\protect\citename{Jovanovic \& Rosenthal, }1988]{jovanovic_1988}
Jovanovic, B., \& Rosenthal, R.~W. 1988.
\newblock Anonymous Sequential Games.
\newblock {\em Journal of Mathematical Economics}, {\bf 17}, 77--87.

\bibitem[\protect\citename{Keymer {\em et~al.\ }\relax, }2000]{keymerMVL2000}
Keymer, Juan~E., Marquet, Pablo~A., Velasco‐Hernández, Jorge~X., \& Levin,
  Simon~A. 2000.
\newblock Extinction Thresholds and Metapopulation Persistence in Dynamic
  Landscapes.
\newblock {\em The American Naturalist}, {\bf 156}(5), 478--494.

\bibitem[\protect\citename{Lachapelle \& Wolfram, }2011]{lachapelleW11}
Lachapelle, Aim\'e, \& Wolfram, Marie-Therese. 2011.
\newblock On a mean field game approach modeling congestion and aversion in
  pedestrian crowds.
\newblock {\em Transportation Research Part B: Methodological}, {\bf 45}(10),
  1572 -- 1589.

\bibitem[\protect\citename{Lasry \& Lions, }2007]{lasry_2007}
Lasry, J.~M., \& Lions, P.~L. 2007.
\newblock Mean Field Games.
\newblock {\em Japanese Journal of Mathematics}, {\bf 2}, 229--260.

\bibitem[\protect\citename{Le~Van \& Stachurski, }2007]{levanS07}
Le~Van, Cuong, \& Stachurski, John. 2007.
\newblock Parametric continuity of stationary distributions.
\newblock {\em Economic Theory}, {\bf 33}(2), 333--348.

\bibitem[\protect\citename{Levin, }1970]{levin70}
Levin, R. 1970.
\newblock Extinction.
\newblock {\em Some mathematical problems in biology. American Mathematical
  Society, Providence, Rhode Island},  77--107.

\bibitem[\protect\citename{Li {\em et~al.\ }\relax, }2017]{LiBPSS2017}
Li, Jian, Bhattacharyya, Rajarshi, Paul, Suman, Shakkottai, Srinivas, \&
  Subramanian, Vijay. 2017.
\newblock Incentivizing Sharing in Realtime D2D Streaming Networks: A Mean
  Field Game Perspective.
\newblock {\em IEEE/ACM Trans. Netw.}, {\bf 25}(1), 3--17.

\bibitem[\protect\citename{Manjrekar {\em et~al.\ }\relax,
  }2014]{manjrekarRS2014}
Manjrekar, M., Ramaswamy, V., \& Shakkottai, S. 2014 (April).
\newblock A mean field game approach to scheduling in cellular systems.
\newblock {\em Pages  1554--1562 of:} {\em IEEE INFOCOM 2014 - IEEE Conference
  on Computer Communications}.

\bibitem[\protect\citename{Mitzenmacher \& Upfal, }2005]{mitzenmacherU05}
Mitzenmacher, Michael, \& Upfal, Eli. 2005.
\newblock {\em Probability and computing: Randomized algorithms and
  probabilistic analysis}.
\newblock Cambridge university press.

\bibitem[\protect\citename{Molofsky, }1994]{molofsky75}
Molofsky, Jane. 1994.
\newblock Population Dynamics and Pattern Formation in Theoretical Populations.
\newblock {\em Ecology}, {\bf 75}(1), 30--39.

\bibitem[\protect\citename{Nelder \& Mead, }1965]{nelder1965simplex}
Nelder, John~A., \& Mead, Roger. 1965.
\newblock A simplex method for function minimization.
\newblock {\em The computer journal}, {\bf 7}(4), 308--313.

\bibitem[\protect\citename{Nisan {\em et~al.\ }\relax,
  }2007]{nisan2007algorithmic}
Nisan, Noam, Roughgarden, Tim, Tardos, Eva, \& Vazirani, Vijay~V. 2007.
\newblock {\em Algorithmic game theory}.
\newblock Cambridge University Press Cambridge.

\bibitem[\protect\citename{Rosenthal, }1973]{rosenthal73}
Rosenthal, Robert~W. 1973.
\newblock A class of games possessing pure-strategy Nash equilibria.
\newblock {\em International Journal of Game Theory}, {\bf 2}(1), 65--67.

\bibitem[\protect\citename{Smith \& McCardle, }2002]{smithM2002}
Smith, James~E., \& McCardle, Kevin~F. 2002.
\newblock Structural Properties of Stochastic Dynamic Programs.
\newblock {\em Operations Research}, {\bf 50}(5), 796--809.

\bibitem[\protect\citename{Weintraub {\em et~al.\ }\relax,
  }2008]{weintraub_2008}
Weintraub, G.~Y., Benkard, C.~L., \& VanRoy, B. 2008.
\newblock Markov Perfect Industry Dynamics with Many Firms.
\newblock {\em Econometrica}, {\bf 76}(6), 1375–--1411.

\bibitem[\protect\citename{Weintraub {\em et~al.\ }\relax,
  }2011]{weintraub_2010}
Weintraub, Gabriel~Y., Benkard, C.~Lanier, \& van Roy, Benjamin. 2011.
\newblock Industry dynamics: Foundations for models with an infinite number of
  firms.
\newblock {\em Journal of Economic Theory}, {\bf 146}(5), 1965 -- 1994.

\bibitem[\protect\citename{Xia, }1994]{xia94}
Xia, Aihua. 1994.
\newblock Weak Convergence of Markov Processes with Extended Generators.
\newblock {\em The Annals of Probability}, {\bf 22}(4), 2183--2202.

\bibitem[\protect\citename{Xu \& Hajek, }2013]{xu2013supermarket}
Xu, Jiaming, \& Hajek, Bruce. 2013.
\newblock The supermarket game.
\newblock {\em Stochastic Systems}, {\bf 3}(2), 405--441.

\end{thebibliography}

\end{document}